\documentclass{article}
\usepackage{graphicx, fullpage, xcolor}
\usepackage{hyperref}
\usepackage{amsmath,amsthm,amsfonts,amssymb}
\usepackage{style_defaults}
\usepackage{subcaption}
\usepackage[shortlabels]{enumitem}
\usepackage{xcolor}
\usepackage{authblk}
\usepackage{cleveref}
\usepackage{multirow}
\usepackage{makecell}
\usepackage{pgfplots}
\usepackage{tikz}
\usetikzlibrary{patterns}
\pgfplotsset{compat=1.18}

\usepackage{algorithm}
\usepackage{algpseudocode}

\usepackage{pifont}

\title{A Principled Approach to Randomized Selection under Uncertainty: Applications to Peer Review and Grant Funding}

\author{Alexander Goldberg}
\author{Giulia Fanti}
\author{Nihar B. Shah}
\affil{Carnegie Mellon University}
\affil{\{\texttt{akgoldbe,gfanti,nihars}\}\texttt{@andrew.cmu.edu}}
\date{}

\newcommand{\nselected}{k}
\newcommand{\nproposal}{n}
\newcommand{\selected}{S}
\newcommand{\perm}{\sigma}

\newcommand{\errorpermutations}{\Sigma_\nproposal}

\newcommand{\prob}{p}
\newcommand{\objvalue}{v}
\newcommand{\lcb}{\ell}
\newcommand{\ucb}{u}
\newcommand{\point}{e}
\newcommand{\nabove}{A}
\newcommand{\nbelow}{B}
\newcommand{\orderwidth}{w}
\newcommand{\maxiters}{T}

\newcommand{\ouralgorithm}{MERIT}

\begin{document}
\maketitle

\begin{abstract}
    Many decision-making processes involve evaluating and then selecting items;  examples include scientific peer review, job hiring, school admissions, and investment decisions. The eventual selection is performed by applying rules or deliberations to the raw evaluations, and then deterministically selecting the items deemed to be the best. These domains often feature error-prone evaluations and uncertainty about future outcomes, which undermine the reliability of such deterministic selection rules. As a result, selection mechanisms involving explicit randomization that incorporate the uncertainty are beginning to gain traction in practice. However, current randomization approaches are ad hoc, and as we prove, inappropriate for their purported objectives. In this paper, we propose a principled framework for randomized decision-making based on interval estimates of the quality of each item. We introduce \ouralgorithm\ (Maximin Efficient Randomized Interval Top-$k$), an optimization-based method that maximizes the worst-case expected number of top candidates selected, under uncertainty represented by overlapping intervals (e.g., confidence intervals or min-max intervals). \ouralgorithm\ provides an optimal resource allocation scheme under an interpretable notion of robustness. We develop a polynomial-time algorithm to solve the associated optimization problem and demonstrate empirically that the method scales to over ten thousand items. Further, we prove that our approach can satisfy desirable axiomatic properties not guaranteed by existing approaches to randomization. Finally, we conduct empirical comparisons of \ouralgorithm\ with existing algorithms on synthetic peer review data based on the Swiss National Science Foundation and NeurIPS and ICLR conferences. Our experiments demonstrate that \ouralgorithm\ matches the performance of existing algorithms in expected utility under fully probabilistic review data models used in previous work, while outperforming previous methods with respect to our novel worst-case formulation.
\end{abstract}

\section{Introduction}
\label{sec:intro}

In many applications like scientific funding, job hiring, school admissions, and startup investment, decision makers evaluate and then select items. Generally, decision makers aim to identify the set of highest-quality candidates from evaluations, possibly after deliberating or trying to account for errors. They then deterministically select these identified top candidates. In such domains, decision makers often grapple with arbitrariness in evaluations and unpredictable outcomes. Recently, there has been growing interest in introducing randomization into the selection process, both to recognize this uncertainty and for a variety of other reasons. In introducing randomization to selection processes, decision makers typically elicit evaluations of candidates through peer review, and then allocate random acceptances by sampling from a probability distribution that depends on the reviews.

Proponents of randomized selection have cited many potential benefits of randomization. Randomization may help in reducing wasted reviewer time adjudicating tie breaks, encouraging high-risk proposals, counteracting ``rich-get-richer'' effects, and combating reviewer partiality \cite{fang2016research, horbach2022partial, heyard2022rethinking, gould2025threats, feliciani2024funding}. Furthermore, surveys have found that a majority of Ph.D.-level scientists support the introduction of randomized selection into peer review-based grant funding \cite{liu2020acceptability, philipps2022research}. Recent work has analyzed how randomization can curb inefficient over-preparation by applicants~\cite{gross2019contest}, align incentives when evaluators may behave strategically~\cite{carnehl2024designing,niemeyer2024optimal}, and allow for causal inference~\cite{azoulay2020scientific} in understanding the usefulness of scientific funding. In fact, many funding agencies have already adopted partial lotteries to allocate grant money, starting with the New Zealand Health Research Council in 2013~\cite{liu2020acceptability}, followed by the Swiss National Science Foundation (NSF) in 2019~\cite{adam2019science}, and recently expanding to numerous agencies around the world~\cite{erc2023peerreview, innovate2024, vwf2025lottery}. There have also been prominent proposals to introduce randomness in college admissions~\cite{adam2019science, klick2024declaration}, job screening~\cite{osterloh2019dealing, berger2020focal} and early stage startup investment~\cite{mckenzie2017identifying}. 

In current deployments, decision makers collect peer-review assessments and then run a lottery where selection probabilities are derived from those review scores. However, the current procedures for this randomization are ad hoc. In this work, we  initiate a principled approach to randomizing decisions from evaluations. We focus on the key question: \emph{Given imperfect evaluations of candidate quality, what is a suitable probability distribution over applicants to make the random selection?} 

We take the perspective of a funder who aims to select the highest quality grant proposals from a set of applicants. A key motivation for randomization in funding decisions is uncertainty about the relative quality of proposals. Indeed, many existing deployments describe peer review lotteries as a form of random ``tie-breaking'' between proposals of equal quality~\cite{NatureEditorial2022Lotteries, SNSF2021Lots}. Therefore, similar to previous work at the Swiss NSF~\cite{heyard2022rethinking}, we assume that the funder estimates numeric quality \emph{intervals} for each proposal; we are agnost ic to the source of these intervals, which could be derived from a number of processes, described in Section \ref{sec:principled_approach}. Concretely, one simple way a funder may generate intervals is by taking the minimum and maximum score given to each interval, instead of simply aggregating all scores to a single estimate like the mean or median. More generally, intervals may capture the funder's uncertainty due to errors in the evaluation process, such as miscalibration and subjectivity, or underlying aleatoric uncertainty about the future success of proposals. Crucially, we capture settings with \emph{``Knightian uncertainty,''} where a funder cannot assign a probability measure over possible outcomes~\cite{knight1921risk}. Models that assume decision makers have Knightian uncertainty are widely applied in related applications such as policy-making~\cite{sunstein2023knightian, ruffino2014knightian, ben2016decision}, financial investment decisions~\cite{epstein2004intertemporal, malenko2020asymmetric, nishimura2007irreversible}, and R\&D investment by businesses~\cite{amoroso2017r}.

The assumption that a funder lacks a probabilistic model of their data is particularly apt when quality is determined by peer review. First, while there are many proposed probabilistic models of reviewer errors, these models have performed poorly in real deployments~\cite[Section `Miscalibration']{langford2012icml, shah2022surveyextended}. One possible reason for this poor performance is that in practice, human miscalibration is more complex than simple models~\cite{brenner2005modeling}. Second, decision makers generally lack ground truth data with which to evaluate whether a given probabilistic model of review data is appropriate in a given setting, making it difficult to rely on a probabilistic model. Further, even if a decision maker had a reasonable model of review scores, these probabilities may not represent meaningful probabilities to the funder, since review scales are arbitrary and likely do not map linearly to the utility of selecting a candidate. Finally, the future success of candidates in the settings we consider may be inherently hard to predict. In the context of scientific peer review, many works find that it is difficult to predict future citations and peer review scores are poor predictors~\cite{schroter2022evaluation, cortes2021inconsistency, weitzner2024predictive}. In the related area of college admissions, recent work demonstrated that predictive models used to rank candidates are highly unstable with respect to their training data, resulting in arbitrariness in the ranking from any single model~\cite{lee2024algorithms}. 

In the absence of a reliable probabilistic model of proposal quality, we assume that  a decision maker describes their uncertainty by estimating intervals of proposal quality rather than point estimates. 
Since we are agnostic to the source of these intervals and do not consider them to have a probabilistic interpretation, the funder draws conclusions from the intervals only about the relative ordering of proposal quality. If two proposals' intervals overlap, then the funder does not have sufficient evidence to conclude that one proposal is higher quality than the other. However, if one interval lies strictly above another, then the funder has sufficient evidence to conclude that one proposal dominates the other. We develop theory and a practical algorithm for making funding decisions on the basis of interval quality estimates, with the following \textbf{key contributions}:  

\begin{enumerate}[leftmargin=*]
    \item \emph{Modeling uncertainty as Knightian uncertainty intervals:}  
     Prior work models a funder's uncertainty by assuming a known probabilistic relationship between proposal quality and review scores~\cite{heyard2022rethinking}, using Bayesian inference to construct confidence intervals. However, we show that in this fully Bayesian setting, a deterministic selection policy always maximizes expected utility, so randomization is unnecessary. In contrast, our model captures the motivation for randomizing without assuming a fully specified probabilistic model, using ``Knightian'' uncertainty intervals, that capture the funder's estimates of the relative quality of proposals. 
    \item \emph{A principled approach to randomization:} We formalize two key principles that guide the design of an algorithm for selection under interval uncertainty and show that randomization is necessary to meet these principles. First, given their uncertainty, the funder aims to solve a maxmin optimization problem, maximizing their worst-case utility over all rankings consistent with quality intervals. Hence, the funder randomizes in order to robustly optimize their worst-case utility. We call this principle \emph{ex ante optimality}. Second, the funder makes a selection in a manner that respects \emph{ex post validity}---if one proposal’s interval lies strictly above another proposal's interval, then the higher quality proposal must be selected whenever the lower quality proposal is selected. Prior heuristic selection rules do not satisfy both of these 
    principles.
    \item \emph{Efficient algorithm}: It is not obvious that it is possible to solve the maximin optimization problem efficiently, as it requires solving a linear program with number of constraints that grows exponentially in the number of selected proposals. In fact, many closely related problems from the literature on (partial order) graph algorithms are NP-hard \cite{faigle1994computational, ravi1991ordering, woeginger2003approximability, adhikary2023complexity}. Nonetheless, we develop a polynomial time algorithm to solve the ex ante optimization problem and enforce ex post validity. We refer to this algorithm as \textbf{M}aximin \textbf{E}fficient \textbf{R}andomized \textbf{I}nterval \textbf{T}op-$k$
    (\textbf{MERIT}). We demonstrate the computational feasibility of \ouralgorithm\ on real-world peer review data from Swiss NSF grant reviews and NeurIPS 2024 and ICLR 2025 conferences. The algorithm runs on a standard personal laptop in under $5$ minutes on inputs with over 10,000 candidates. We release an open source implementation of our \ouralgorithm\ algorithm at \href{https://github.com/akgoldberg/lottery}{\tt github.com/akgoldberg/lottery}.

    \item \emph{Axiomatic comparison}: Since ground truth of candidate quality tends to be unavailable in the domains we consider, we initiate an axiomatic approach to comparing randomized mechanisms. We identify three desirable properties of ``monotonicity in budget'', ``stability'', and ``reversal symmetry.'' We prove that \ouralgorithm\ prevents instances of ``maximal instability'' and respects ``reversal symmmetry''  while the existing randomization mechanisms studied do not. We show that all randomization mechanisms that we consider  
    violate monotonicity and propose modifications to \ouralgorithm\ to enforce this axiom; however, these modifications cause us to violate ex ante optimality.
    \item \emph{Empirical comparison:} We evaluate \ouralgorithm\ against existing methods using synthetic data based on real peer review data from major conferences (NeurIPS 2024, ICLR 2025) and grant agencies (Swiss NSF 2020). \ouralgorithm\ performs comparably to existing methods in expected utility under a linear reviewer error model used by the Swiss NSF~\cite{heyard2022rethinking} and many other prior works~\cite{flach2010kdd,baba2013quality,roos2011calibrate,roos2012statistical}. However, under our worst-case objective---which assumes any ordering consistent with quality intervals could be true—our algorithm significantly outperforms deterministic selection and the Swiss NSF's randomized approach.
\end{enumerate}

The remainder of this paper is organized as follows. In Section~\ref{sec:problem_setup}, we overview current deployments of randomized decisions in scientific funding and motivate our proposed new approach. In Section~\ref{sec:problem_defn}, we formally define our framework for randomized selection. In Section~\ref{sec:algorithm}, we develop an efficient randomized selection algorithm that optimally solves the robust optimization problem over intervals defined in our framework. In Section~\ref{sec:axioms}, we propose axioms that randomized selection rules should satisfy and theoretically compare methods with respect to these axioms. In Section~\ref{sec:empirical_results}, we simulate  scientific peer review selection processes using synthetic data and publicly available peer review datasets to empirically compare \ouralgorithm\ and existing methods on these datasets. Finally, we conclude with related work in Section~\ref{sec:related_work} and discussion of limitations and directions for future work in Section~\ref{sec:discussion}. For clarity of presentation, formal proofs are deferred to Appendix~\ref{app:proofs} throughout the main text.

\section{Background and Approach}
\label{sec:problem_setup}

We begin by describing current deployments of randomized decisions in scientific funding. We survey existing approaches and highlight their drawbacks in Section~\ref{sec:existing_deployments} and then motivate our proposed framework in Section~\ref{sec:principled_approach}.

\subsection{Existing Deployments}
\label{sec:existing_deployments}

In recent years, there have been many deployments of ``peer review lotteries'' in scientific funding decisions. Most deployments use an approach of \emph{``randomize-above-threshold.''} As described in Algorithm~\ref{alg:rand_above threshold}, the funder chooses a minimum acceptable quality threshold and samples uniformly at random among all proposals that are above this threshold. This approach has been adopted to make funding decisions by the European Research Council~\cite{erc2023peerreview}, the British Academy~\cite{britishacademy2025randomisation}, the Science Foundation of Ireland~\cite{innovate2024}, the Volkswagen Foundation~\cite{vwf2025lottery}, and the Health Research Council of New Zealand~\cite{liu2020acceptability} among others.  Additionally, the USENIX Security Conference is randomly allocating long oral presentations among accepted papers to their conference in 2025~\cite{usenix2025policy}. However, as we describe in Section~\ref{sec:problem_defn}, randomize-above-threshold may violate a desired principle of ``ex post validity'', which says that if one proposal clearly dominates another, the stronger proposal should be funded if the weaker proposal is funded.

\begin{algorithm}[ht]
\caption{Randomize Above Threshold~\cite{liu2020acceptability, erc2023peerreview, innovate2024, vwf2025lottery, britishacademy2025randomisation, usenix2025policy}}
\label{alg:rand_above threshold}
\begin{algorithmic}[1]
\Statex \textbf{Input:} Set of proposals with point estimates and intervals; number accepted $\nselected$.
\Statex \textbf{Output:} Set of proposals $\nselected$.
\Statex 
\State Choose a threshold $T$ (potentially based on the data).
\State Reject all intervals strictly below $T$ and select uniformly at random among the remaining. 
\end{algorithmic}
\end{algorithm}

Taking a different approach, the \emph{Swiss National Science Foundation (NSF)}~\cite{heyard2022rethinking, trentacosti2021lotteries} pioneered a method that explicitly accounts for uncertainty about the quality of each proposal. 
They assume that each proposal has a latent true quality and review scores are generated based on these quality scores and reviewer-specific noise parameters. 
They assume priors on the model parameters and develop methods to obtain point estimates and confidence intervals for the true quality of each proposal. As described in Algorithm~\ref{alg:swiss_nsf}, the Swiss NSF then samples $\nselected$ proposals by setting a ``provisional funding line'' as the $\nselected$-th highest point estimate. All proposals with intervals strictly above the funding line are selected and all proposals strictly below the funding line are rejected. The remaining budget is allocated uniformly at random among proposals with intervals that overlap the funding line.

\begin{algorithm}[ht]
\caption{Swiss NSF Selection Algorithm~\cite{heyard2022rethinking}}
\label{alg:swiss_nsf}
\begin{algorithmic}[1]
\Statex \textbf{Input:} Set of proposals with point estimates and intervals of the quality of each proposal; number accepted $\nselected$.
\Statex \textbf{Output:} Set of proposals $\nselected$.
\Statex 
\State Rank proposals by decreasing point estimate and let $\point_{(\nselected)}$ be the point estimate of the $\nselected$-th ranked proposal.
\State Let $\mathcal{A}$ be the set of proposals with lower bound strictly above $\point_{(\nselected)}$, $\mathcal{R}$ be the set of proposals with upper bound strictly below $\point_{(\nselected)}$, and $\mathcal{P}$ the set of proposals with intervals that contain $\point_{(\nselected)}$.
\State Accept $\mathcal{A}$, reject $\mathcal{R}$, and accept ($\nselected - |\mathcal{A}|$) proposals chosen uniformly at random from $\mathcal{P}$.
\end{algorithmic}
\end{algorithm}

The intuition provided by the Swiss NSF for their approach is that confidence intervals capture the funder's uncertainty in estimating the proposal quality and randomizing decisions accounts for this uncertainty, thereby leading to better decisions. In contrast, as we prove in \Cref{prop:fully_bayesian_model} below, when the funder assumes that review data is generated from a fully specified Bayesian model, there exists a deterministic selection of proposals that maximizes the funder's expected utility for any utility function. Informally, \emph{if the funder knows the model that generates their data, then they do not need to randomize in order to maximize their expected utility.} Below, we present the formal proposition and proof. We prove a general statement for a funder who estimates a total ranking of proposals and receives a pre-specified utility for any possible estimate of the ranking of proposals and true qualities of the proposals. The top-$\nselected$ selection problem is a case of this estimation problem where the funder's utility specifies the utility of choosing $\nselected$ proposals based on the estimated ranking of proposals.

\begin{proposition}[Optimality of deterministic selection in the fully Bayesian setting]
    \label{prop:fully_bayesian_model}
    Consider a funder who estimates a ranking of proposals from review data in the following fully Bayesian setting. The funder observes review data $y \in Y$ for a set of proposals of true quality $\theta \in X$, for some measurable sets $X$ and $Y$. The review data and true quality are generated jointly from a known probability distribution. Letting $\Pi$ denote the set of all permutations of proposals, the funder estimates a ranking $\hat{\pi}$ of proposals and gains utility $u(\hat{\pi}, \theta)$  where $u: \Pi \times X \to \R$.
     The funder aims to choose a (potentially randomized) estimator of the true ranking $f: Y \to \Delta(\Pi)$ that maximizes their expected utility $\E_{y, \theta, \hat{\pi}}[ u(\hat{\pi}, \theta)]$ where $\hat{\pi} \sim f(y)$ . 
   In this setting, there always exists a deterministic $f$ that maximizes the funder's expected utility. 
\end{proposition}

\Cref{prop:fully_bayesian_model} suggests one drawback of the Swiss NSF's method---in their model setting there may be a utility cost to randomization compared to choosing deterministically in an optimal manner. A second drawback is that the Swiss NSF's algorithm for selecting proposals from intervals violates natural axioms for a selection rule, like monotonicity in the budget $\nselected$ and stability, as we show in Section~\ref{sec:axioms}. 

\subsection{Our Approach}
\label{sec:principled_approach}

In our work, we propose a model that captures the motivation for randomizing due to uncertainty about the relative quality of the proposals. We show that if the funder cares about their \emph{worst case utility} they must randomize decisions in order to robustly optimize their utility. 

Specifically, we consider a funder who estimates intervals for each proposal based on data. These intervals need not come from any one particular model, but they should capture the funder's inherent uncertainty about the relative quality of proposals.

Our interpretation of the quality intervals stems from the intuitive argument that if the intervals for two proposals overlap then the funder does not have enough evidence to distinguish between them. On the other hand, if the interval of proposal $A$ dominates the interval of proposal $B$, then the funder has sufficient evidence to believe that $A$ is better than $B$. Hence, the intervals define a partial ordering of proposals that represents a set of conclusions by the funder regarding the relative quality of different proposals. This ordering is the canonical ``interval order'' for a set of intervals. It captures the spirit of how the Swiss NSF interprets confidence intervals in their setup, as they describe that the intervals are \emph{``used to identify proposals with similar quality and $\ldots$ these proposals are entered into a lottery to select those to be funded. The approach acknowledges that there are proposals of similar quality and merit, which cannot all be funded.''}~\cite{heyard2022rethinking}

Our approach only considers the relative ordering of the intervals and not additional information such as the amount of overlap between intervals, point estimates, or other information from a probabilistic model of the data. This makes our approach particularly applicable to settings where decision-makers have so-called ``Knightian uncertainty'', that is, when they cannot assign a single probability measure over the set of possible outcomes~\cite{knight1921risk, sunstein2023knightian}. As further motivation, we describe a number of concrete settings in which a funder would construct intervals under this Knightian uncertainty:

\begin{enumerate}[leftmargin=*]
    \item \emph{Imputation-based intervals:} in many grant funding panels most proposals have been reviewed by a large fraction of the reviewers. For example, in the Swiss NSF process~\cite{heyard2022rethinking}, over $91\%$ pairs of reviewer-proposal pairs received scores. If every reviewer scored every proposal, then the funder need not worry about reviewer miscalibration---the tendency for different reviewers to interpret the review scale differently\cite{wang2018your}. Hence, the funder may first impute values for missing scores and then aggregate across reviewers. The funder can make minimal assumptions by adopting Manski bounds~\cite{manski1990nonparametric, saveski2023counterfactual}. Manski bounds generate \emph{intervals} by imputing missing scores with a range of possible values from the minimum to maximum score. The aggregate scores are then given as an interval over the range of possible imputed values. Hence, the intervals represent the plausible set of values for each proposal's aggregate review score, without a probabilistic interpretation. 
    \item \emph{Intervals from model ensembling:} in many cases, a funder may have a number of plausible models of their data, each of which produces point estimates. Hence, the funder can estimate all plausible models of the data and ensemble into intervals by taking quantiles of the point estimates. This type of ensembling has been applied to time-series forecasting problems where it is  known as Quantile Prediction Averaging~\cite{nowotarski2015computing}.
    \item \emph{Intervals based on expert input:} a frequently cited motivation for randomization is concerns about prejudice, for example, against highly original ideas or junior researchers. The funder may not be able to reliably estimate such sources of error from their observational, potentially sparse data, but can rely on prior controlled experiments that establish the rough magnitude of prejudices in the review process.
    \item \emph{Multi-criteria aggregation:} Reviewers are often given a number of criteria on which proposals are rated, for example ``intellectual merit'' and ``broader impact.'' The funder can consider multiple valid ways to aggregate 
    criteria and generate intervals over the set of possible aggregations. 
    \item \emph{Robustness to mis-specification of model used to generate intervals:} the funder may use a probabilistic model to estimate intervals and then draw conclusions about the ordering of the intervals, but may not trust other information about the distribution. This is a common assumption in the literature on distributionally robust optimization (see~\cite{Rahimian_2022} for a survey), where the funder is said to be optimizing over a ``support-only'' ambiguity set.
\end{enumerate}

\section{Problem Formulation}
\label{sec:problem_defn}

Our methods apply to settings like admissions, scientific peer review, job screening, and financial investment, where decision makers estimate quality intervals and select top candidates based on these intervals. For concreteness, throughout our exposition, we will describe a \emph{funder} choosing \emph{proposals}.

Consider a funder who receives $\nproposal$ proposals. From these, the funder wishes to select the $\nselected$  highest quality proposals. Note that $\nproposal$ could be as large as thousands of proposals and $\nselected$ a fixed fraction of the total and can also be in the hundreds or thousands. For each proposal $i \in [\nproposal]$\footnote{We use the standard notation $[\kappa]$ to denote set $\{1,\ldots,\kappa\}$ for any positive integer $\kappa$.}, the funder estimates an interval $[\lcb_i, \ucb_i] \subseteq \R$ representing a range of quality scores that the proposal could possibly take. A higher score indicates higher quality. The funder wishes to design a randomized selection mechanism to choose  $\nselected$ proposals given the intervals. In order to design such a mechanism, we adopt two primary principles which we term as ex ante optimality and ex post validity, described below. 

\paragraph{Ex ante optimality}
The funder's utility is the expected number of the true top-$\nselected$ proposals that they select, ranked by their quality. Formally, let $\perm: [\nproposal] \to [\nproposal]$ denote a ranking of the proposals where for each proposal $i \in [\nproposal]$, the rank of the proposal is denoted by $\perm(i) \in [\nproposal]$. If the funder samples proposals with marginal probabilities $\prob \in [0,1]^\nproposal$ and the true ranking of proposals is $\perm$, their expected utility is $\sum_{i=1}^\nproposal \prob_i \ind\{\perm(i) \leq \nselected\}$. Clearly, if the funder knew the true ranking $\perm$, then they would optimize utility by choosing deterministically, i.e., by setting $\prob_i=1$ for $i$ with $\perm(i) \leq \nselected$, and $0$ otherwise.

However, recall that the funder is uncertain about relative qualities of the proposals, as captured by overlaps in intervals. Any ordering that is consistent with overlapping intervals could be the true ranking. Hence, the intervals define a set of feasible rankings:
    \[
    \errorpermutations = \left\{ \perm \text{ permutation of } [\nproposal] \mid \forall i, j \in [\nproposal],\; \lcb_i > \ucb_j \implies \perm(i) < \perm(j) \right.\}
    \]
In other words, if proposal $i$ has quality strictly above proposal $j$, then $i$ is ranked higher than $j$ in all $\perm \in \errorpermutations$.
As an example of an extreme case, if all $\nproposal$ intervals overlap each other, then $\errorpermutations$ consists of all possible permutations of the proposals. 

 For ex ante optimality, the funder optimizes their $\emph{worst case}$ utility over feasible rankings $\errorpermutations$:
\begin{align}
 \min_{\perm \in \errorpermutations} \sum_{i = 1}^{\nproposal} \prob_i \ind\{\perm(i) \leq \nselected\} \;\;
\label{ref:ex_ante_objective}
\end{align}
The funder maximizes their worst-case expected utility by choosing the optimal marginal probabilities $\prob$ solving the maximin optimization problem:
    \begin{align}
        \max_{\substack{\prob \in [0,1]^\nproposal: \\ \|\prob\|_1 = \nselected}} \;   \min_{\perm \in \errorpermutations} \; \sum_{i = 1}^{\nproposal} \prob_i \ind\{\perm(i) \leq \nselected\}. 
        \label{ref:opt_problem}
    \end{align}
Finally, the funder randomly chooses $\nproposal$ proposals with marginal probabilities corresponding to $\prob$. 

The ex ante optimization problem has a game theoretic interpretation that motivates the need for randomization. Our model corresponds to a zero-sum Stackelberg game where the funder is the ``leader'' who selects $\nselected$ proposals. The funder faces an adversarial ``follower'' who chooses a ranking of proposals. The leader's utility is the number of top $\nselected$ proposals selected based on the adversary's ranking, while the adversary's utility is the negation of the leader's. The Strong Stackelberg Equilibrium (SSE) is exactly the solution to Objective~\ref{ref:opt_problem}. It is well known that in an SSE, the leader may need to commit to a randomized (or mixed) strategy.

\paragraph{Ex post validity}
The ex post validity criterion requires that for any pair of proposals $a$ and $b$, if $b$'s quality interval lies strictly below $a$'s interval and if $b$ is selected, then $a$ must also be selected.

Formally, a selection rule that takes as input a set of quality intervals $I$ and outputs a set of selected proposals $\selected$ satisfies \emph{ex post validity}, if for all pairs of intervals $a,b \in I$ with $\lcb_a > \ucb_b$, and all outputs $\selected$ selected with non-zero probability, $b \in \selected \implies a \in \selected$.

The ex post validity criteria ensures that the actual selected set of proposals is legitimate to stakeholders. In particular, if the funder rejects a proposal that dominates an accepted proposal, that would be unacceptable to the funder and to applicants. 

While the ex post condition seems natural, the simple randomize-above-threshold mechanism can violate it: Suppose proposals $a$ and $b$ both lie above the threshold, but $a$ dominates $b$. Because $a$ and $b$ are entered into a uniform lottery, $a$ may be rejected at random, while $b$ is accepted at random.

\section{Efficient Algorithm}
\label{sec:algorithm}

The ex ante optimization problem in (\ref{ref:opt_problem}) is equivalent to to solving the following linear program (LP):
\begin{align}
\label{LP} 
    \max_{\prob \in \mathbb{R}^n, \objvalue \in \mathbb{R}} \quad & \objvalue  \\ 
    \text{subject to} \quad 
    & \objvalue \leq \sum_{i=1}^{\nproposal} \prob_i \ind\{\perm(i) \leq \nselected\}, \quad \forall \perm \in \errorpermutations, \notag \\
    & \sum_{i=1}^{\nproposal} \prob_i = \nselected \text{ and } 0 \leq \prob_i \leq 1, \forall i \in [\nproposal] \notag 
\end{align}
This LP can have on the order of $\binom{\nproposal}{\nselected}$ constraints so its size is exponential in $\nselected$. Recall that in practice, $\nproposal$ can be thousands of proposals, out of which the funder selects $\nselected$ in the many hundreds or thousands, so this LP is intractably large. In fact, determining feasibility of a solution to the LP when $\errorpermutations$ is defined by an arbitrary partial order, called the \emph{minimum weight $k$-ideal problem}, is NP-hard~\cite{faigle1994computational}. Prior work has shown that the problem is solvable in polynomial time for a tree partial order, but our problem uses an interval order. Several algorithmic problems related to interval orders are also known to be NP hard, including maximum cut on interval graphs~\cite{adhikary2023complexity}, interval graph completion~\cite{ravi1991ordering} and minimizing weighted completion time of jobs subject to interval ordering precedence constraints~\cite{woeginger2003approximability}.

In contrast, as we show in Section~\ref{sec:separation_oracle}, the ex ante optimization problem is solvable in polynomial time using the ellipsoid method with a separation oracle. We present a practical cutting plane algorithm based on the theoretically polynomial time algorithm in Section~\ref{sec:practical_algorithm}. Finally, in Section~\ref{sec:ex_post_algo}, we describe how to ensure that random sampling with optimal marginal probabilities satisfies both ex ante and ex post conditions simultaneously, with an efficient algorithm. We call this end-to-end algorithm, summarized in Section~\ref{sec:full_algo}, \textbf{M}aximin \textbf{E}fficient \textbf{R}andomized \textbf{I}nterval \textbf{T}op-$\nselected$
 (\textbf{MERIT}).

\subsection{Polynomial Time Algorithm}
\label{sec:separation_oracle}

We now develop a polynomial-time algorithm to solve linear program~\eqref{LP}. Our approach solves the problem using a polynomial time ``separation oracle''~\cite{grotschel1981ellipsoid} with the ellipsoid algorithm. A separation oracle checks whether a proposed solution satisfies all the constraints.
If the solution is feasible, the oracle confirms it. If not, it identifies (at least one) specific constraint that the solution violates. The separation oracle may be used to solve the LP without enumerating all (exponentially many) constraints by starting with a limited set of constraints and iteratively shrinking the possible feasible region of the LP through calls to the separation oracle. Our main result proves that this method yields a polynomial time algorithm:

\begin{theorem}[Polynomial time solution] The linear program~\eqref{LP} can be solved within accuracy $\epsilon$ of the optimal solution in polynomial time with respect to $\nproposal$ and $\log(1/\epsilon)$ using the ellipsoid algorithm with Algorithm~\ref{alg:separationoracle} as a separation oracle.
\label{prop:ellipsoid}
\end{theorem}

The primary technical difficulty is the design of an efficient separation oracle. We present our separation oracle in Algorithm~\ref{alg:separationoracle}. Given a candidate solution $(\prob, \objvalue)$, the separation oracle checks whether $\prob$ achieves a worst-case objective value of at least $\objvalue$. If the worst-case objective value under $\prob$ is greater than $\objvalue$, then the solution is feasible, if not the oracle returns a set of violated constraints. At a high level, the oracle works by constructing worst-case possible sets of top-$\nselected$ proposals. For each of the $\nselected + 1$ intervals with the largest lower bounds, the algorithm constructs the worst-case set of top-$\nselected$ proposals that includes intervals $1$ to $(i-1)$ and excludes interval $i$ in the top-$\nselected$. Excluding an interval with a large lower bound constrains the set of intervals that must be in the top-$\nselected$, since all intervals strictly below the interval with the large lower bound must be excluded from the top-$\nselected$. 
In considering all such sets of intervals, the algorithm enumerates possible worst-case permutations with respect to $\prob$ in time $O(\nproposal \nselected)$. If the separation oracle finds a permutation that gives objective values smaller than $\objvalue$, it returns this permutation, which corresponds to a violated constraint in the LP. If it does not find any such permutation, then $(\prob, \objvalue)$ is feasible. This separation oracle is used as a sub-routine in the ellipsoid algorithm to compute the optimal solution in polynomial time in Theorem~\ref{prop:ellipsoid}.  

\begin{lemma}[Polynomial-time separation oracle]
    \label{proposition:separation_oracle}
For any candidate solution $(\prob, \objvalue)$ to the linear program~\eqref{LP}, Algorithm~\ref{alg:separationoracle} returns $\emptyset$ only if the candidate solution is feasible and returns a non-empty set of violated constraints otherwise. Moreover, Algorithm~\ref{alg:separationoracle} runs in time $O(\nproposal \max\{\nselected, \log \nproposal\})$. 
\end{lemma}

\newcommand{\cuts}{Z}
\begin{algorithm}
\caption{Polynomial-time Separation Oracle}
\begin{algorithmic}[1]
\Statex \textbf{Input:} Candidate solution: $(\prob, \objvalue) \in [0,1]^\nproposal \times \R$ with $\|\prob\|_1 = \nselected$, number selected $\nselected$, set of intervals $\{\lcb_i, \ucb_i\}_{i \in [\nproposal]}$ sorted in decreasing order of lower bound $\ell_i$
\Statex \textbf{Output:} A set of violated constraints ($\emptyset$ if $(\prob, \objvalue)$ is feasible)
\State $\cuts \gets \emptyset$
\For{$i = 1$ to $\nselected+1$}
     \State $S_i \gets \{j \in (i, \nproposal]: \text{intervals } j \text{ and } i \text{ overlap} \}$
     \If{$|S_i| \geq (\nselected - (i-1))$}
        \State Obtain $\tilde S_i$ by sorting $S_i$ by $\prob$ and keeping only the $\nselected- (i-1)$ smallest values
        \If{$\objvalue > \sum_{j=1}^{i-1} \prob_j + \sum_{j \in \tilde S_i } \prob_j \;^\dagger$
            \State $\cuts \gets \cuts \cup \left\{\text{``\ }\objvalue \leq \sum_{j=1}^{i-1} \prob_j + \sum_{j \in \tilde S_i} \prob_j\text{''\ }\right\}$
            }
        \EndIf
     \EndIf
\EndFor
\State \Return $\cuts$
\Statex
\Statex $^\dagger$\textit{By convention, we take the empty sum from $j=1$ to $0$ to be $0$.}
\end{algorithmic}
\label{alg:separationoracle}
\end{algorithm}

\subsection{Practical Algorithm}
\label{sec:practical_algorithm}

While theoretically polynomial time, in practice, the ellipsoid algorithm is impractical for problems with even a few hundred proposals. Hence, taking inspiration from the ellipsoid algorithm, we develop a practical ``cutting plane'' method to solve the LP~\eqref{LP}. 

The cutting plane algorithm is described in full in Algorithm~\ref{alg:cutting_plane}.
The algorithm starts by solving a relaxation of the LP without any of the worst-case value constraints to find an initial (potentially infeasible) candidate solution $(\prob, \objvalue)$. Then, the algorithm repeatedly calls the separation oracle to check the feasibility of the current candidate solution. If the candidate solution is feasible, it is an optimal solution to the LP, since it is optimal for a relaxation of the full LP. If the candidate solution is infeasible, the algorithm adds the constraints returned by the separation oracle and re-solves the LP.

\begin{algorithm}[ht]
\caption{Cutting Plane Algorithm}
\begin{algorithmic}[1]
\Statex \textbf{Input:} Number of proposals to select $\nselected$, set of intervals $I=\{\lcb_i, \ucb_i\}_{i \in [\nproposal]}$ max iterations $\maxiters$ 
\Statex \textbf{Output:} Ex ante optimal vector of marginal probabilities $\prob$
\Statex
\Statex \textit{\# Prune Intervals}
\State For all intervals strictly below at least $\nselected$ other intervals, set $\prob_i = 0$ and remove. 
\State For all intervals strictly above at least $\nproposal - \nselected$ other intervals, set $\prob_i = 1$ and remove.
\State Let $a$ be the number of intervals accepted. Update $\nselected \gets \nselected - a$.
\Statex
\Statex \textit{\# Initialize Linear Program}
\State For each interval $i \in [\nproposal]$ compute $A(i)$, the number of proposals strictly above $i$ and $B(i)$, the number of proposals strictly below $i$.
\State Using $A$ and $B$, partition the intervals into $\orderwidth$ monotonically ordered subsets $M_1,\ldots,M_\orderwidth$. (Algorithm~\ref{alg:chain-cover})
\State Solve the following linear program to obtain initial $\prob, \objvalue$:
\begin{align*}
\min_{\objvalue, \prob} \quad & \objvalue \\
\text{s.t.} \quad & \sum_{i=1}^{\nproposal} \prob_i = \nselected \\
                  & \prob_i \in [0,1] \quad \forall i \in [\nproposal] \\
                  & \objvalue \leq \sum_{j=1}^{\nselected} \prob_{j} \\
                  & \prob_{M[i]} \geq \prob_{M[i+1]} \quad \forall i \in [|M|-1], \forall M \in \{M_1, \ldots, M_\orderwidth\}
\end{align*}
\Statex \textit{\# Add Cuts}

\For{$\maxiters$ iterations}

\State $C \gets \text{\texttt{SeparationOracle}}((\prob, \objvalue), \nselected, I) $
\If{$C = \emptyset$} \Comment{Feasible solution}
\State \Return \prob
\Else  \Comment{Infeasible solution}
\State Add constraints from $C$ to the LP and solve the new problem to obtain new $(\prob, \objvalue)$.
\EndIf
\EndFor
\State \Return \text{\texttt{Failure}}
\end{algorithmic}
\label{alg:cutting_plane}
\end{algorithm}

The cutting plane algorithm converges to a feasible optimal solution quickly in practice because it is initialized with a useful set of constraints on the feasible region of the problem. These constraints prune the problem and impose monotonicity and symmetry constraints on the marginal probabilities $\prob$, based on the number of intervals above and below each proposal, which we define below. 

\begin{definition}[Number above ($\nabove$) and number below ($\nbelow$)]
\label{defn:nabovenbelow}
For each proposal $i \in [\nproposal]$, define:
    \begin{align*}
        \nabove(i) = |\{r \in [\nproposal]: \lcb_r > \ucb_i \}| \qquad \text{and} \qquad  
        \nbelow(i) = |\{r \in [\nproposal]: \lcb_i > \ucb_r \}|,
    \end{align*}
that is, $\nabove(i)$ is the number of intervals strictly above $i$ and $\nbelow(i)$ is the number of intervals strictly below $i$.
\end{definition}

The cutting plane algorithm first prunes all intervals guaranteed to always be in the top $\nselected$ or never be in the top $\nselected$, which reduces the number of decision variables. Then, the cutting plane algorithm initializes the linear program with a set of constraints on the optimal marginal probabilities $\prob$. In particular, the algorithm partitions the intervals into ``monotonically ordered subset'' within which $\prob$ can be assumed to be monotonically non-increasing without loss of optimality:

\begin{definition}[Monotonically ordered subset]
    \label{defn:monotonically_ordered_subsets}
    For any subset of intervals $M \subseteq [\nproposal]$, we say that the subset of intervals is \emph{monotonically ordered} if $\forall i \in [|M|-1]$, $\nabove(M[i]) \leq \nabove(M[i+1])$ and $\nbelow(M[i]) \geq \nbelow(M[i+1])$, where $M[i]$ denotes the $i$-th element in $M$.
\end{definition}

The monotonicity constraints come from the observation that if proposal $a$ has more proposals below it than proposal $b$ and has fewer proposals above it than proposal $b$, then without loss of optimality, $\prob_a \geq \prob_b$. Our cutting plane algorithm therefore imposes constraints on these ``monotonically ordered subsets'' at the start, which greatly constrains the size of the feasible region, while still returning an optimal solution. We provide detailed analysis of the correctness and efficiency of this algorithm in Appendix~\ref{app:full_cutting_plane}.

\subsection{Enforcing Ex Post Validity}
\label{sec:ex_post_algo}

A solution to the ex ante optimality LP \eqref{LP} returned by the Cutting Plane Algorithm (Algorithm~\ref{alg:cutting_plane}) or the Ellipsoid Algorithm, is not guaranteed to output a vector of marginal probabilities, such that sampling proposals with these marginals always guarantees ex post validity. However, we prove that we can post-process any solution to the ex ante optimization problem, and then sample with marginal probabilities $\prob$ to guarantee ex ante optimality and ex post validity simultaneously. This stands in contrast to the commonly used ``randomize-above-threshold'' approach to randomization, which does not guarantee ex post validity as described in Section~\ref{sec:problem_defn}.  

\begin{theorem}[Post-processing for ex post validity] Given any ex ante optimal $\prob$, Algorithm~\ref{alg:tie_breaking} enables the funder to sample $\nselected$ proposals while satisfying both ex ante and ex post conditions and is computable in time $O(\nproposal^2)$.
\label{thm:ex_post}
\end{theorem}

Theorem~\ref{thm:ex_post} applies the post-processing algorithm given in Algorithm~\ref{alg:tie_breaking} to a solution from the Cutting Plane Algorithm. For any $a, b \in [\nproposal]$ with $\lcb_a > \ucb_b$, Algorithm~\ref{alg:tie_breaking} terminates with $\prob_a = 1$ or $\prob_b = 0$. Moreover, Algorithm~\ref{alg:tie_breaking} never decreases the objective value of $\prob$. Hence, applying post-processing to an ex ante optimal solution is without loss of optimality and ensures that any sampling method that selects proposals with marginal probabilities $\prob$ satisfies ex post validity. 

Finally, there are many possible methods to implement this sampling (see \cite{tille2023_sampling} for a survey). We use the simple and classical ``systematic sampling''~\cite{madow1949theory} with runtime $O(\nproposal)$. This method requires only one uniform random sample to implement, so the randomness could be easily documented for transparency, which is an important part of the current Swiss NSF procedure~\cite{trentacosti2021lotteries} where a lottery is conducted by publicly drawing papers from a glass jar. We provide pseudo-code in Appendix~\ref{app:systematic_sampling}.

\begin{algorithm}
\caption{Post-Processing of $\prob$ for Ex Post Validity}
\label{alg:tie_breaking}
\begin{algorithmic}[1]
\Statex \textbf{Input:} Vector of marginal probabilities $\prob$, sequence of intervals $\{[\lcb_i, \ucb_i]\}_{i \in [\nproposal]}$
\Statex \textbf{Output:} Vector of marginal probabilities $\prob$
\State Order the intervals by increasing $\ucb$.
\For{$b \in [\nproposal]$}
    \If{$\prob_b = 0$}
        \State continue
    \EndIf
    \For{$a$ from $\nproposal$ to $(b+1)$}
        \If{$\lcb_a > \ucb_b$ and $\prob_a < 1$}
            \State $d \gets \min\{\prob_b, 1 - \prob_a\}$
            \State $\prob_b \gets \prob_b - d$
            \State $\prob_a \gets \prob_a + d$
        \EndIf
    \EndFor
\EndFor
\end{algorithmic}
\end{algorithm}

\subsection{Full Algorithm}
\label{sec:full_algo}

The complete algorithm \ouralgorithm\ solves the ex ante optimization procedure with post-processing for ex post validity (Algorithm~\ref{alg:tie_breaking}) followed by sampling. The algorithm is provably polynomial time if the ex ante optimization problem is solved using the ellipsoid method with a separation oracle as proved in Theorem~\ref{prop:ellipsoid}. In practice, we solve the optimization problem using the cutting plane algorithm given in Algorithm~\ref{alg:cutting_plane}, which yields a practically efficient algorithm, albeit with a non-polynomial time theoretical convergence guarantee, described in Appendix~\ref{app:full_cutting_plane}.

\begin{algorithm}
\caption{\ouralgorithm\ Algorithm}
\label{alg:metric}
\begin{algorithmic}[1]
\Statex \textbf{Input:} Number of proposals to select $\nselected$, set of intervals $I=\{\lcb_i, \ucb_i\}_{i \in [\nproposal]}$
\Statex \textbf{Output:} Selection of $\nselected$ proposals
\State Compute an ex ante optimal vector of marginal probabilities $\prob$ using Algorithm~\ref{alg:cutting_plane}.
\State Apply ex post validity post-processing to $\prob$ (Algorithm~\ref{alg:tie_breaking}).
\State Sample $\nselected$ proposals from $[\nproposal]$ with marginal probabilities of inclusion given by $\prob$ (Algorithm~\ref{alg:systmetaic_sampling}).
\end{algorithmic}
\end{algorithm}

\section{Axiomatic Comparison}
\label{sec:axioms}

In many of the applications we consider, like scientific funding or college admissions, there is no agreed upon ground-truth measurement of the quality of selections made by the decision maker. Hence, it is unclear how to empirically measure whether one algorithm performs better than another algorithm in selecting ``top'' candidates. Therefore, we take inspiration from social choice theory~\cite{brandt2016handbook} and initiate an axiomatic comparison of our \ouralgorithm\ method with the alternative methods discussed in Section~\ref{sec:existing_deployments}---deterministic top-$\nselected$ selection, randomize above threshold, and the Swiss NSF method. We consider a number of natural properties, or ``axioms'', that a selection algorithm should obey and theoretically analyze the behavior of \ouralgorithm\ and alternative algorithms with respect to these axioms.

\subsection{Defining Axioms}

We propose three natural desiderata of algorithms for selecting proposals from quality assessments. We begin by defining a generic ``randomized selection rule'' which we will characterize axiomatically.

\newcommand{\qualityinput}{I}
\newcommand{\proboutput}[2]{\prob(#2, #1)}

\begin{definition}[Selection rule]
    A \emph{selection rule} receives as input a set of $\nproposal$ quality estimates $\qualityinput = \{(\lcb_i, \point_i, \ucb_i)\}_{i \in [\nproposal]}$ where $\lcb_i, \ucb_i$ are lower and upper limits on the quality of the item $i$ and $\point_i \in [\lcb_i, \ucb_i]$  is a point estimate of the quality. Given a budget $\nselected \in \{1,\ldots, \nproposal\}$, the selection rule outputs a subset of $[\nproposal]$ of size $\nselected$. We let $\proboutput{\nselected}{\qualityinput} \in [0,1]^\nproposal$ denote the vector of marginal probabilities of selecting each item:
\end{definition}

This definition captures methods that do not actually use intervals to make decisions like deterministic selection, methods that use only intervals like \ouralgorithm, and methods that use both intervals and point estimates like the Swiss NSF's approach. Next, we define three axioms that a selection rule ought to obey.

First, we propose the axiom of \textbf{``monotonicity in budget''.} This axiom requires that increasing the number of proposals selected, $\nselected$, should not decrease the probability of selection for any proposal. Formally:
\begin{definition}[Monotonicity in budget]
    \label{defn:monotonicity}
    A selection rule respects monotonicity in budget, if for any input $\qualityinput$ and for all budgets $\nselected \in [\nproposal-1]$, it must be that $\proboutput{\nselected+1}{\qualityinput}_i \geq \proboutput{\nselected}{\qualityinput}_i \; \forall i \in [\nproposal]$.
\end{definition}

Second, an important property for a selection rule is ``stability'' --- changing one interval should not lead to a large change in the behavior of the algorithm. We define an instance of instability that is undesirable for an algorithm. We consider two ``extreme'' behaviors for an algorithm. At one extreme, an algorithm behaves deterministically, selecting $\nselected$ proposals with probability $1$, which is the minimum entropy distribution from which to sample. At the other extreme, the algorithm samples uniformly at random among all $\nproposal$ proposals, which is the maximum entropy distribution. We say that an algorithm exhibits \textbf{``maximal instability''}, if changing a single interval by an arbitrarily small amount can change the algorithm's behavior between these two extremes. Formally:

\begin{definition}[Maximum instability]
    \label{defn:instability} A selection rule is \emph{maximally unstable}, if there exist a pair of inputs $I$ and $J$ that differ by an arbitrarily small amount $\epsilon > 0$ in the quality estimate of a single proposal and a budget $\nselected \in \{2,\ldots,\nproposal - 2\}$ such that $\proboutput{\nselected}{I} = \frac{\nselected}{\nproposal} \mathbf{1}_\nproposal$ (the mechanism samples uniformly at random among all proposals on input $I$) whereas $\proboutput{\nselected}{J} \in \{0,1\}^\nproposal$ (the mechanism selects deterministically on input $J$).
\end{definition}

A selection algorithm should \emph{avoid} maximum instability. We note that in the definition of instability, we restrict the budget to lie in $\{2,\ldots,\nselected-2\}$, because stability with respect to changing a single proposal is not meaningful when choosing only one proposal to accept or reject. In particular, the top-$1$ or bottom-$1$ proposal can change entirely based on a change to a single proposal. 

Finally, inspired by the \textbf{``reversal symmetry''} axiom from social choice theory~\cite{saari2012geometry}, we define a natural notion of symmetry that says that when selecting $1$ out of $2$ proposals, if the scale of quality is reversed so that all intervals are flipped, then the selection rule should flip the probabilities of selection for the two proposals:

\begin{definition}[Reversal symmetry]
    For any input $\qualityinput = \{(\lcb_i, \point_i, \ucb_i)\}_{i \in [\nproposal]}$ where $\lcb_i, \point_i, \ucb_i \in [0,1] \forall i \in [\nproposal]$, let $\qualityinput^{(R)}= \{(1-\ucb_i, 1-\point_i, 1-\lcb_i)\}_{i \in [\nproposal]}$ be the horizontally flipped (or reversed) input. Then, a selection rule that selects $\nselected = 1$ out of $\nproposal = 2$ proposals respects reversal symmetry if for any pair of flipped inputs $\qualityinput$ and $\qualityinput^{(R)}$, $\prob(\qualityinput, 1) = (\prob_1, \prob_2)$ and $\prob(\qualityinput^{(R)}, 1) = (\prob_2, \prob_1)$.
\end{definition}

\subsection{Theoretical Analysis}
We now present our theoretical results on which algorithms satisfy the three axioms defined. Deterministic top-$\nselected$ selection meets both criteria, but does not account for the funder's uncertainty (our earlier ex ante requirement). We characterize randomized selection mechanisms with respect to these axioms in Theorem~\ref{thm:axiomatic_analysis}:

\begin{theorem}[Axiomatic analysis]
\label{thm:axiomatic_analysis}
Existing randomized algorithms have the following properties:
\begin{enumerate}[label=(\alph*), leftmargin=*]
    \item Swiss NSF and randomize-above-threshold both violate maximum instability, while \ouralgorithm\ is never maximally unstable.
    \item Swiss NSF, randomize-above-threshold and \ouralgorithm\ all violate monotonicity in budget.
    \item It is not possible to simultaneously satisfy ex ante optimality and monotonicity in budget.
    \item Swiss NSF and randomize-above threshold violate reversal symmetry, while \ouralgorithm\ satisfies reversal symmetry.
\end{enumerate}
\end{theorem}

We give a formal proof of Theorem~\ref{thm:axiomatic_analysis} in Appendix~\ref{ref:axiom_proofs}. Here we provide examples where the three axioms fail . First, Figure~\ref{fig:monotonicity_failure} shows an example where none of the randomized selection mechanisms considered respect monotonicity in budget. In particular, when $\nselected=1$, both Swiss NSF and \ouralgorithm\ sample uniformly at random between $1$ and $2$, so proposal $2$ is selected with probability $\frac{1}{2}$. However, when $\nselected=2$, proposal $1$ is always in the top $2$ and the other top-$2$ proposal could be any one of $2$, $3$, or $4$, so both algorithms choose proposal $1$ deterministically and sample uniformly among $2,3,4$ meaning proposal $2$ is selected with probability $\frac{1}{3}$. Hence, perhaps counterintuitively, proposal $2$ is worse off when the budget $\nselected$ increases from $2$ to $3$. We note that randomize-above-threshold behaves identically to Swiss NSF here, if the threshold is chosen in a data-dependent manner as the $\nselected$-the highest point estimate, so this same example shows that randomize-above-threshold violates monotonicity.

The example in Figure~\ref{fig:monotonicity_failure}, also demonstrates claim (3) of Theorem~\ref{thm:axiomatic_analysis} that it is not possible to simultaneously satisfy ex ante optimality and monotonicity in budget. In particular, the unique ex ante optimal $\prob$ for this example are not monotonic in budget. As we discuss in Section~\ref{sec:enforcing_monotonicity}, it is possible to modify \ouralgorithm\ to enforce monotonicity in the budget at the potential cost of ex ante optimality.

Interestingly, in reproducing the Swiss NSF's selection process using real 2020 grant review data, we find that both the Swiss NSF selection rule and \ouralgorithm\ exhibit non-monotonicity in budget in practice. For example, using the Swiss NSF method, there are proposals whose probability of selection \emph{decreases} from $0.75$ to $0.524$ when increasing budget $\nselected$ from $89$ to $90$. This suggests that violations of the monotonicity in budget axiom can occur in real-world deployments.

Second, in Figure~\ref{fig:max_instability_failure}, we show an example that violates maximum instability for Swiss NSF and (data-dependent) randomize-above threshold. In this example, where $\nselected=3$, slightly decreasing the upper bound and point estimate of interval $3$, changes the algorithm's behavior from choosing the top-$3$ deterministically, to sampling uniformly at random among all proposals. In contrast, \ouralgorithm\ never exhibits maximal instability (for $\nselected \in [2, \nproposal-2]$). Our proof of this is in Theorem~\ref{thm:axiomatic_analysis}, follows by observing that \ouralgorithm\ only samples uniformly at random among all proposals if all intervals overlap, and will not select deterministically for any set of intervals where only a single interval does not overlap the remaining intervals.

\begin{figure}[ht]
    \centering
    \includegraphics[width=0.3\linewidth]{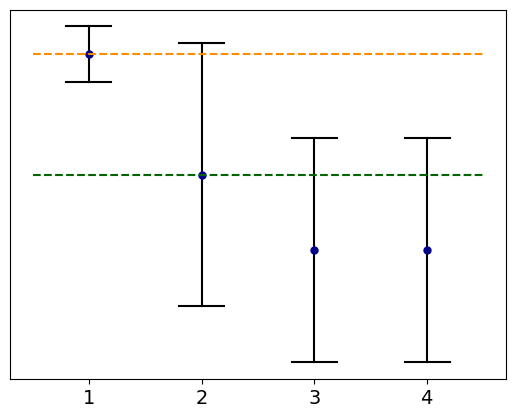}
    \caption{Example that violates \emph{monotonicity} with respect to $\nselected$ for Swiss NSF and our \ouralgorithm\ algorithm. When $\nselected=1$, $\prob_2 = 1/2$ for both algorithms. However, when $\nselected=2$, $\prob_2 = 1/3$ for both algorithms.}
    \label{fig:monotonicity_failure}
\end{figure}

\begin{figure}[ht]
    \centering
    \includegraphics[width=0.7\linewidth]{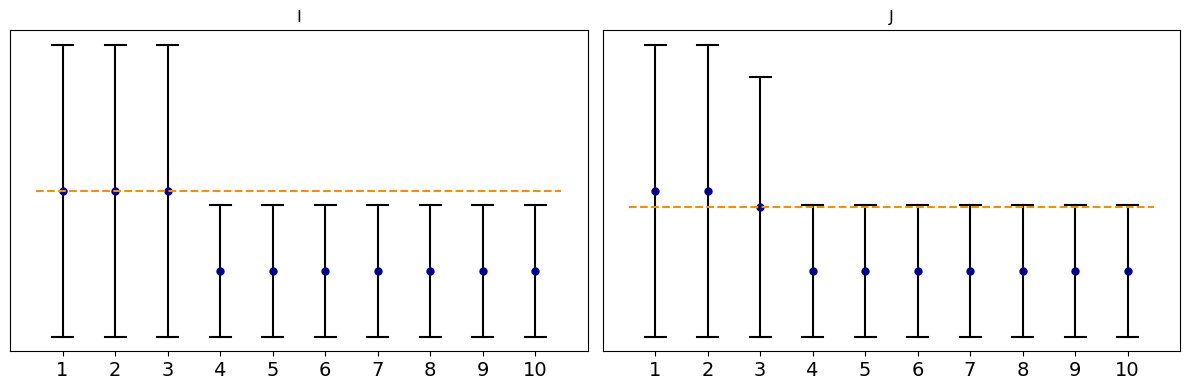}
    \caption{Example that violates \emph{maximum instability} for Swiss NSF and randomize-above-threshold, with $\nselected=3$. Slightly decreasing the upper bound and point estimate of proposal $3$, changes the algorithm's behavior from selecting the top $3$ deterministically (left) to sampling among all $10$ proposals uniformly at random (right).}
    \label{fig:max_instability_failure}
\end{figure}

Finally, an example where reversal symmetry is violated by the Swiss NSF method is when selecting one proposal from two intervals with values $(0,1)$ and $(0.1, 0.2)$ and point estimates at the midpoints of $0.5$ and $0.15$. Then, the Swiss NSF selection rule will accept interval $1$ and reject interval $2$. However, if quality is flipped, then proposal $1$ remains the same, but proposal $2$ has quality interval $(0.8, 0.9)$ with point estimate of $0.85$. Then, for the reversed instance, Swiss NSF samples uniformly at random between the two proposals, violating reversal symmetry. In contrast, \ouralgorithm\ samples uniformly at random between the two for both the original and reversed intervals.

\subsection{Enforcing Monotonicity in Budget}
\label{sec:enforcing_monotonicity}

In this section, we show that it is possible to modify \ouralgorithm\ to enforce monotonicity in budget, by solving the ex ante optimization problem sequentially for a budget ranging from $1$ to the desired $\nselected$ with an additional constraint to the optimization problem that the sequence of $\prob$ should be non-decreasing. At each step, we solve the optimization problem and then apply post-processing for ex post validity using Algorithm~\ref{alg:tie_breaking}. This may come at the cost of ex ante optimality, but always enforces monotonicity in budget and ex post validity as we prove in Proposition~\ref{prop:monotonicity_algo}.

\begin{algorithm}[h]
\caption{Solving Sequence of Optimization Problems for Monotonicity}
\label{alg:monotonicity}
\begin{algorithmic}[1]
\Statex \textbf{Input:} Sequence of intervals $\{[\lcb_i, \ucb_i]\}_{i \in [\nproposal]}$, number selected $\nselected$.
\Statex \textbf{Output:} Sequence of marginal selection probabilities $\prob^{(1)},\ldots,\prob^{(\nselected)}$ for selecting $1$ to $\nselected$ proposals.
\State Let $\prob^{(0)} = (0,\ldots,0)$
\For{$i \in [\nselected]$}
\State Compute ex ante optimal $\prob$ for selecting top-$i$ proposals using Algorithm~\ref{alg:cutting_plane} with the additional constraint $\prob^{(i)}_j \geq \prob^{(i-1)}_j \; \forall j \in [\nproposal]$ to obtain optimal $\prob^{(i)}$. 
\State Post-process $\prob^{(i)}$ with Algorithm~\ref{alg:tie_breaking} for ex post validity.
\EndFor
\State \Return $\prob^{(1)},\ldots,\prob^{(\nselected)}$.
\end{algorithmic}
\end{algorithm}

\begin{proposition}[Monotonicity in budget $\nselected$]
\label{prop:monotonicity_algo}
Algorithm~\ref{alg:monotonicity} guarantees that $\prob^{(i)} \geq \prob^{(i-1)}$ for all $i \in [\nselected]$ and any algorithm that samples $i$ proposals with marginal probabilities $\prob^{(i)}$ respects ex post validity.     
\end{proposition}

\section{Experiments}
\label{sec:empirical_results}

In order to compare the performance of different selection algorithms, we use review data where we assume known proposal quality (or intervals containing the true quality). The funder observes error-prone estimates of quality---both interval estimates and point estimates---and then applies a selection rule. We present head-to-head comparisons of different selection algorithms under both the linear miscalibration model used by the Swiss NSF and other prior works, and under our worst-case model in Section~\ref{sec:empirical_comparison}. We then describe computational runtime experiments demonstrating the efficiency of our algorithm in \ouralgorithm\ in Section~\ref{sec:computation_experiments}. Finally, we present qualitative insights into how different methods may lead to different outcomes in real peer review settings in Section~\ref{sec:case_studies}.

\subsection{Comparison of Methods in Expected Utility and Worst-Case Utility}
\label{sec:empirical_comparison}

We evaluate the performance of selection algorithms from two perspectives: (1) expected utility in a probabilistic model of reviewer behavior, and (2) worst-case utility over feasible rankings defined by score uncertainty intervals. While our algorithm, \ouralgorithm, is designed for worst-case optimality, it is informative to assess its expected performance under plausible generative models of the review process. Defining a realistic probability distribution over reviewer scores is challenging, but one model that has been widely adopted in the literature is a linear miscalibration model, which underlies the interval-generation method of the Swiss National Science Foundation (Swiss NSF)~\cite{heyard2022rethinking} and has been used in many other peer review settings~\cite{flach2010kdd,baba2013quality,roos2011calibrate,roos2012statistical,tan2021least}.

\subsubsection*{Expected Utility under a Linear Miscalibration Model}

\newcommand{\sdquality}{\sigma_\theta}
\newcommand{\sdmiscal}{\sigma_b}
\newcommand{\sderror}{\sigma_\epsilon}

A popular probabilistic model used by the Swiss NSF and many others captures error in the review process due to \emph{reviewer miscalibration}, the tendency for some reviewers to be stricter or more lenient than others. The model is linear, where each proposal has a numeric true quality and each reviewer has a miscalibration offset, both drawn from Gaussian distributions. The review score for a reviewer on a specific paper is a linear combination of true quality and reviewer offset plus Gaussian noise. This model has been applied widely to peer review data in prior work~\cite{flach2010kdd,baba2013quality,roos2011calibrate,roos2012statistical}. The model is specified by an assignment of reviewers to proposals and three parameters for the standard deviation of the quality ($\sdquality$), reviewer miscalibration ($\sdmiscal$), and per-review error ($\sderror$).

We generate synthetic data from this model in two regimes based on real-world settings: (1) Swiss NSF's grant review process with 350 proposals, 10 reviewers, and 80 proposals per reviewer and (2) Computer Science conference peer review with 1000 reviewers, 1000 papers, and 5 papers per reviewer. In both cases, we assign papers to reviewers uniformly at random and assume a score range of $1$ to $10$ with $\sdquality = 2$, $\sdmiscal=1$ and $\sderror=0.5$, matching parameters used in prior work~\cite{tan2021least}. In both settings, the funder aims to select the top one-third of proposals, reflecting realistic acceptance rates. 

We generate 50 synthetic samples of review datasets. Then, following the methodology of the Swiss NSF, estimate model parameters and posterior expected rank of each proposal along with 50\% confidence intervals for the expected rank given the observed review data. These point estimates and intervals are given as input to each selection rule. We test both a deterministic rule that selects the top-$\nselected$ by mean score, denoted \emph{Deterministic (mean)}, and a deterministic rule that selects the top-$\nselected$ by estimated expected rank based on the model, denoted \emph{Deterministic (model)}.

\subsubsection*{Our Worst-Case Objective}

Our worst-case objective \eqref{ref:ex_ante_objective}, assumes that any ordering consistent with quality intervals could be the true ranking of proposals. The utility is then defined as the worst-case expected fraction of top-$\nselected$ proposals chosen over all possible orderings consistent with the intervals. We simulate generating such quality intervals in three scientific peer review scenarios and measure the performance of different selection methods with respect to our worst-case objective. First, we replicate the scientific grant funding process of the Swiss NSF. We use publicly released review data from their 2020 grant review process consisting of $\nproposal=353$ proposals~\cite{heyard2022rethinking}. We generate intervals using their linear model of reviewer miscalibration. Additionally, we generate imputation-based intervals for the Swiss NSF data using Manski bounds to impute missing reviewer-proposal scores and aggregating scores across reviewers using the median. Second, we use paper review data from the NeurIPS 2024 conference, available on OpenReview, which includes $\nproposal=4035$ accepted papers. We simulate a process of allocating long talks (orals and spotlights) among accepted papers, inspired by the USENIX Security 2025~\cite{usenix2025policy}. Third, we use paper review data from the ICLR 2025 conference, available on OpenReview, which includes all $\nproposal=11520$ papers submitted to the conference. We simulate allocating paper acceptances at this conference. For NeurIPS and ICLR, there is no standard method to generate intervals. We therefore implement three different methods for generating intervals on NeurIPS and ICLR data:
\begin{enumerate}[label=(\arabic*), leftmargin=*]
    \item \emph{Leave-one-out intervals (LOO)}: taking inspiration from ``jacknife'' or leave-one-out intervals~\cite{barber2021predictive}, we compute the range of possible mean review scores, leaving out one reviewer at a time for each paper.
    \item \emph{Gaussian error model credibility intervals}: similar to the Swiss NSF's intervals, we assume a Gaussian model generates review scores based on underlying true quality scores and infer credibility intervals for the true quality. Specifically, for paper $i$, we assume paper has true quality $\theta_i \sim \mathcal{N}(0, 2)$, precision $\tau_i \sim \text{Gamma}(1,1)$, and review scores on the paper are drawn i.i.d. from $\mathcal{N}(\theta_i, 1/\sqrt{\tau_i})$. The parameters of the priors are chosen to closely match those of the Swiss NSF model. We infer $50\%$ credibility intervals for $\theta_i$ given the observed review scores using MCMC.
    \item \emph{Subjectivity intervals}: NeurIPS 2024 and ICLR 2025 both asked reviewers to provide numerical scores of papers' soundness, presentation, and contribution in addition to overall scores. Previous works have observed that different reviewers may have different subjective views of which criteria matter to a paper's quality introducing arbitrariness into overall review scores~\cite{lee2015commensuration}. One proposed approach to mitigate this subjectivity is to learn a mapping from sub-criteria to an overall score based on peer review data~\cite{noothigattu2021loss} and then use this mapping to adjust the review score. We generate intervals by applying this method to adjust scores and taking the interval to be all values in between original scores and subjectivity adjusted scores.
\end{enumerate}

We note that there are no confidence intervals for estimates generated in the worst-case setting as these are generated for a single dataset with a single set of intervals, where there is no randomness in the data generation process.

\subsubsection*{Results}

\begin{figure}[htb]
  \centering

  \begin{tikzpicture}
    \begin{axis}[
        hide axis,
        xmin=0, xmax=1,
        ymin=0, ymax=1,
        width=0pt,
        height=0pt,
        scale only axis,
        legend style={
            at={(0.5,0.5)},
            anchor=center,
            legend columns=4,
            legend cell align=left,
            draw=none,
            fill=none,
            font=\small,
            column sep=1em
        },
        area legend
    ]
    \addplot+[ybar, draw=black, fill=red!60!black!30, postaction={pattern=crosshatch, pattern color=black}] coordinates {(-1,-1)};
    \addplot+[ybar, draw=black, fill=orange!60, postaction={pattern=horizontal lines, pattern color=black}] coordinates {(-1,-1)};
    \addplot+[ybar, draw=black, fill=green!50!black!30, postaction={pattern=grid, pattern color=black}] coordinates {(-1,-1)};
    \addplot+[ybar, draw=black, fill=blue!40, postaction={pattern=north east lines, pattern color=black}] coordinates {(-1,-1)};
    \legend{Deterministic (mean), Deterministic (model), Swiss NSF, MERIT}
    \end{axis}
  \end{tikzpicture}
  
  \vspace{0.5em} 
  
  \begin{subfigure}[b]{0.42\textwidth}
    \centering
    \begin{tikzpicture}
    \begin{axis}[
        ybar,
        bar width=0.25cm,
        width=5.2cm,
        height=6cm,
        enlarge x limits=0.5,
        ymin=0.8, ymax=0.95,
        ylabel={Expected Utility},
        ylabel style={font=\small},
        symbolic x coords={Swiss NSF, Conference},
        xtick=data,
        xtick style={draw=none},
        x tick label style={rotate=45, anchor=east, font=\small},
        ytick style={draw=none},
        tick label style={font=\small}
    ]
    \addplot+[
        ybar,
        draw=black,
        fill=red!60!black!30,
        postaction={pattern=crosshatch, pattern color=black},
        error bars/.cd,
        y dir=both,
        y explicit
    ] coordinates {
        (Swiss NSF, 0.844) +- (0.014955, 0.014955)
        (Conference, 0.889) +- (0.014955, 0.014955)
    };
    \addplot+[
        ybar,
        draw=black,
        fill=orange!60,
        postaction={pattern=horizontal lines, pattern color=black},
        error bars/.cd,
        y dir=both,
        y explicit,
        error bar style={black}
    ] coordinates {
        (Swiss NSF, 0.919) +- (0.004479, 0.004479)
        (Conference, 0.934) +- (0.004479, 0.004479)
    };
    \addplot+[
        ybar,
        draw=black,
        fill=green!50!black!30,
        postaction={pattern=grid, pattern color=black},
        error bars/.cd,
        y dir=both,
        y explicit
    ] coordinates {
        (Swiss NSF, 0.919) +- (0.004427, 0.004427)
        (Conference, 0.933) +- (0.004427, 0.004427)
    };
    \addplot+[
        ybar,
        draw=black,
        fill=blue!40,
        postaction={pattern=north east lines, pattern color=black},
        error bars/.cd,
        y dir=both,
        y explicit
    ] coordinates {
        (Swiss NSF, 0.918) +- (0.004438, 0.004438)
        (Conference, 0.934) +- (0.004438, 0.004438)
    };
    \end{axis}
    \end{tikzpicture}
    \caption{Linear miscalibration model.}
    \label{fig:expected_utility_miscalibration}
  \end{subfigure}
  \hfill
  \begin{subfigure}[b]{0.56\textwidth}
    \centering
    \begin{tikzpicture}
    \begin{axis}[
        ybar,
        bar width=0.25cm,
        width=10.64cm,
        height=6cm,
        enlarge x limits=0.15,
        ylabel={Worst-case Utility},
        ylabel style={font=\small},
        symbolic x coords={
            Swiss NSF,
            NeurIPS LOO, NeurIPS Gauss, NeurIPS Subj,
            ICLR LOO, ICLR Gauss, ICLR Subj
        },
        xtick=data,
        xtick style={draw=none},
        x tick label style={rotate=45, anchor=east, font=\footnotesize},
        ytick style={draw=none},
        tick label style={font=\small},
        ymin=0, ymax=1
    ]
    \addplot+[
        ybar,
        draw=black,
        fill=red!60!black!30,
        postaction={pattern=crosshatch, pattern color=black}
    ] coordinates {
        (Swiss NSF, 0.934)
        (NeurIPS LOO, 0.054)
        (NeurIPS Gauss, 0.199)
        (NeurIPS Subj, 0.377)
        (ICLR LOO, 0.647)
        (ICLR Gauss, 0.579)
        (ICLR Subj, 0.650)
    };
    \addplot+[
        ybar,
        draw=black,
        fill=green!50!black!30,
        postaction={pattern=grid, pattern color=black}
    ] coordinates {
        (Swiss NSF, 0.939)
        (NeurIPS LOO, 0.371)
        (NeurIPS Gauss, 0.231)
        (NeurIPS Subj, 0.347)
        (ICLR LOO, 0.637)
        (ICLR Gauss, 0.574)
        (ICLR Subj, 0.680)
    };
    \addplot+[
        ybar,
        draw=black,
        fill=blue!40,
        postaction={pattern=north east lines, pattern color=black}
    ] coordinates {
        (Swiss NSF, 0.943)
        (NeurIPS LOO, 0.395)
        (NeurIPS Gauss, 0.367)
        (NeurIPS Subj, 0.540)
        (ICLR LOO, 0.690)
        (ICLR Gauss, 0.684)
        (ICLR Subj, 0.749)
    };
    
    
    \draw[red, thick] (axis cs:NeurIPS Gauss, 0.231) -- (axis cs:NeurIPS Gauss, 0.367);
    \draw[red] ([xshift=-2pt] axis cs:NeurIPS Gauss, 0.231) -- ([xshift=2pt] axis cs:NeurIPS Gauss, 0.231);
    \draw[red] ([xshift=-2pt] axis cs:NeurIPS Gauss, 0.367) -- ([xshift=2pt] axis cs:NeurIPS Gauss, 0.367);
    \node[anchor=south, font=\scriptsize, fill=white, inner sep=1pt] at (axis cs:NeurIPS Gauss, 0.390) {+0.136};
    
    \draw[red, thick] (axis cs:NeurIPS Subj, 0.347) -- (axis cs:NeurIPS Subj, 0.540);
    \draw[red] ([xshift=-2pt] axis cs:NeurIPS Subj, 0.347) -- ([xshift=2pt] axis cs:NeurIPS Subj, 0.347);
    \draw[red] ([xshift=-2pt] axis cs:NeurIPS Subj, 0.540) -- ([xshift=2pt] axis cs:NeurIPS Subj, 0.540);
    \node[anchor=south, font=\scriptsize, fill=white, inner sep=1pt] at (axis cs:NeurIPS Subj, 0.565) {+0.193};
    
    \draw[red, thick] (axis cs:ICLR Gauss, 0.574) -- (axis cs:ICLR Gauss, 0.684);
    \draw[red] ([xshift=-2pt] axis cs:ICLR Gauss, 0.574) -- ([xshift=2pt] axis cs:ICLR Gauss, 0.574);
    \draw[red] ([xshift=-2pt] axis cs:ICLR Gauss, 0.684) -- ([xshift=2pt] axis cs:ICLR Gauss, 0.684);
    \node[anchor=south, font=\scriptsize, fill=white, inner sep=1pt] at (axis cs:ICLR Gauss, 0.710) {+0.110};
    
    \end{axis}
    \end{tikzpicture}
    \caption{Worst-case over interval ordering model (our model).}
    \label{fig:worst_case_utility}
  \end{subfigure}
  \caption{Proportion of top-$\nselected$ proposals selected by different methods with quality data generated under the Swiss NSF model of linear miscalibration and under our model of worst-case over feasible rankings. MERIT matches performance of algorithms designed for the Swiss NSF's linear model, with expected utility averaged over 50 samples of synthetic data and error bars showing 95\% CI for the sample mean. The gap between MERIT and other methods in the worst-case over intervals defined by our model can be substantial, as shown by the gaps in NeurIPS Gaussian, NeurIPS Subjectivity, and ICLR Subjectivity.}
  \label{fig:comparison_both_models}
\end{figure}
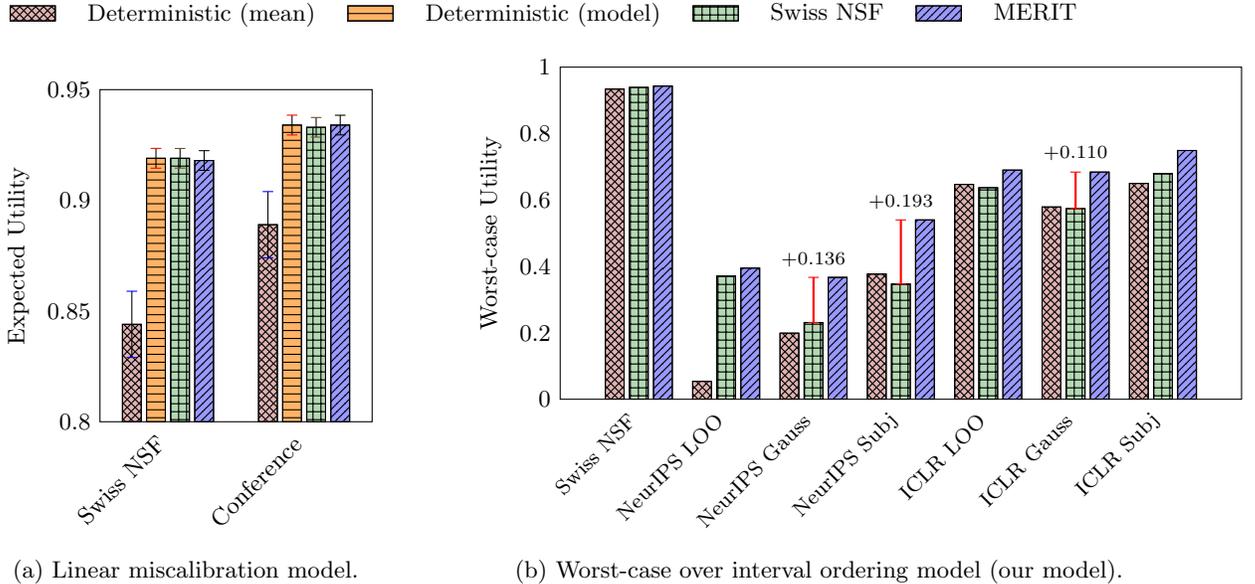

In Figure~\ref{fig:comparison_both_models} we show results of our simulations. Under the Swiss NSF's fully Bayesian linear model, \ouralgorithm\ performs similarly in expectation to the Swiss NSF's randomized method and to deterministic selection using model-adjusted scores. Using the raw mean scores yields much lower performance, as this method does not account for the miscalibration error. In contrast, in the worst-case setting of our model, \ouralgorithm\ outperforms other methods on all baselines. This is expected as \ouralgorithm\ optimizes for this objective. Interestingly, our experiments, demonstrate that the optimality gap can be quite large, with \ouralgorithm\ achieving utility of up to 0.19 more than Swiss NSF, the other randomized algorithm considered. Deterministic selection sometimes performs quite poorly, for instance on NeurIPS LOO and NuerIPS Gaussian, when intervals are relatively wide. Even though Swiss NSF's selection rule randomizes based on the intervals, the optimality gap between Swiss NSF and our ex ante optimal \ouralgorithm\ algorithm can also be large, for example on NeurIPS Gaussian and NeurIPS Subjectivity data, where Swiss NSF does not perform much better than deterministic selection, while \ouralgorithm\ has higher utility. Taken together, these results suggest that \ouralgorithm\ does not degrade expected utility in the fully Bayesian setting, while offering additional robustness with respect to our worst-case objective. 

We additionally conduct ablations on parameters of both models that provide insight into the relative performance of different methods under varying model settings. First, in Figure~\ref{fig:linear_miscal}, we show utility of different methods in the Swiss NSF's linear miscalibration model varying the degree of miscalibration. As the magnitude of miscalibration increases, deterministic selection using mean score degrades greatly in performance it does not account for error due to miscalibration at all. Meanwhile, the other methods perform similarly and maintain fairly high expected utility, even with a large degree of miscalibration. We observe similar results in additional ablations of the miscalibration model and when selecting the top one-tenth of proposals instead of top one-third, as shown in Appendix~\ref{app:additional_data}. 

In Figure~\ref{fig:manski}, we test Manski bounds on the Swiss NSF dataset of grant reviews, where we impute missing values with the full range of scores. We artificially drop reviews to increase the sparsity of review scores, leading to worse utility. With 40\% sparsity, deterministic achieves near zero utility as almost all intervals overlap. At all sparsity levels, \ouralgorithm\ outperforms both Swiss NSF and Deterministic selection. In fact, at sparsity of 0.15 to 0.25, the optimality gap between Swiss NSF and \ouralgorithm\ is the same as that of deterministic selection and \ouralgorithm.

\begin{figure}[tb]
\centering
\begin{minipage}[t]{0.48\textwidth}
    \centering
    \vspace{0pt} 
    \begin{subfigure}[t]{\textwidth}
        \centering
        \includegraphics[width=0.95\linewidth]{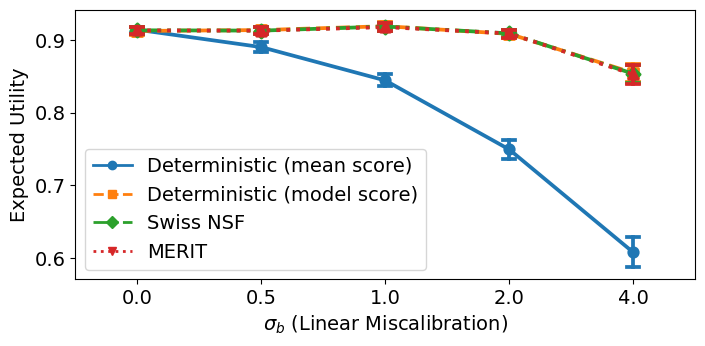}
        \caption{Swiss NSF simulated data.}
        \label{fig:linear_miscal_swissnsf}
    \end{subfigure}    
    \begin{subfigure}[t]{\textwidth}
        \centering
        \includegraphics[width=0.95\linewidth]{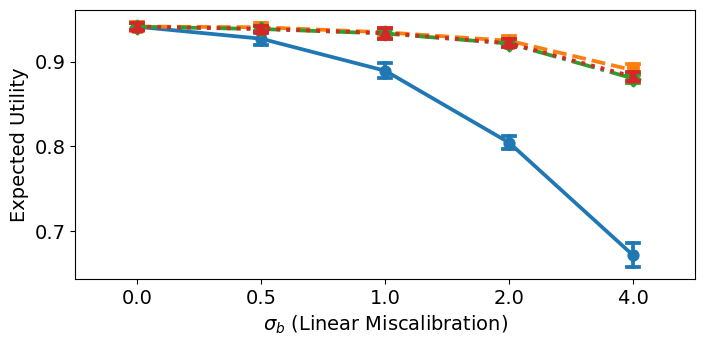}
        \caption{Conference simulated data.}
        \label{fig:linear_miscal_conference}
    \end{subfigure}
    
    \caption{Ablation study comparing methods under the Swiss NSF's model of linear miscalibration with varying levels of miscalibration. Error bars show bootstrapped 95\% CIs for the sample mean over 50 samples of randomly generated data from the model.}
    \label{fig:linear_miscal}
\end{minipage}
\hfill
\begin{minipage}[t]{0.48\textwidth}
    \centering
    \vspace{0pt} 
    \includegraphics[width=0.95\linewidth]{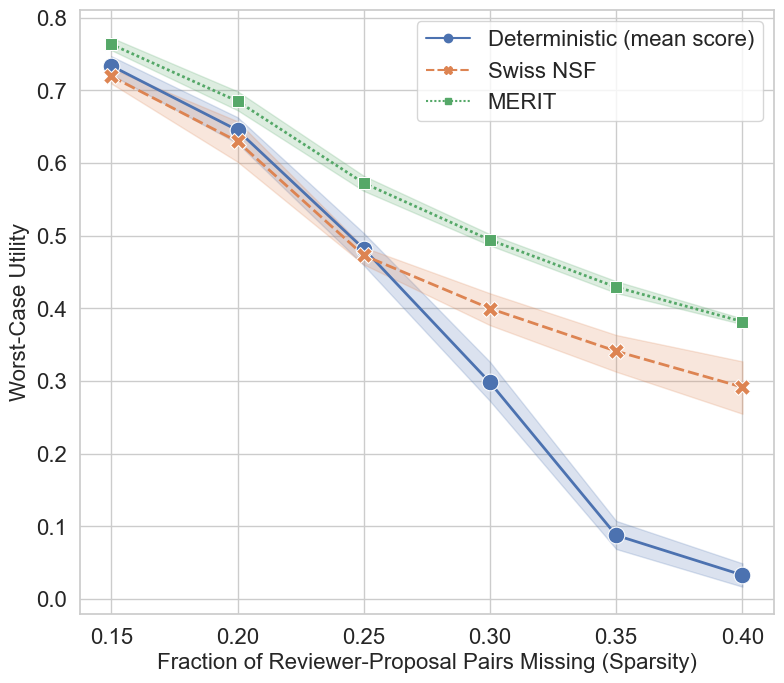}
    
    \caption{Worst-case utility over Manski bound intervals as a function of the fraction of reviewer-proposal pairs missing, with random dropping of review scores to increase sparsity for the Swiss NSF panel review dataset. Error bars show bootstrapped 95\% CIs for the sample mean over 50 trials of randomly dropping review scores.}
    \label{fig:manski}
\end{minipage}
\end{figure}

\subsection{Computational Efficiency}
\label{sec:computation_experiments}

All experiments are run locally on a standard personal laptop (a 2019 MacBook Pro with a 2.6 GHz 6-core Intel Core i7 processor) without any specialized hardware. We solve linear programs using Gurobi~12.0.1~\cite{gurobi}, which solves LPs using the simplex method. We find that for all methods \ouralgorithm\ runs in under five minutes, even with over 11,000 proposals. We provide further analysis on runtime, size of the LP solved, and convergence rate in Appendix~\ref{app:additional_data_computation} including details on variations in runtime as a function of different datasets and interval generation methods. 

\subsection{Comparison of Outcomes in Case Studies}
\label{sec:case_studies}

In order to give insight into how differing approaches may lead to different types of lotteries in actual peer review settings, we present qualitative differences between \ouralgorithm\ and the Swiss NSF approach.

In Table~\ref{tab:probability_results}, we provide a high-level comparison of the marginal probabilities of the sampling proposals given under \ouralgorithm\ and the Swiss NSF algorithm. One simple comparison point is on the number of proposals that entered into a lottery ($\%$ Random) under each algorithm. We find that \ouralgorithm\ tends to randomize over an equal or greater number of candidates than the Swiss NSF algorithm---on Swiss NSF data, NeurIPS LOO, NuerIPS Gaussian, and ICLR Subjectivity the two algorithms randomize over a similar number of candidates, while on NeurIPS Subjectivity, ICLR LOO, and ICLR Gaussian, \ouralgorithm\ randomizes over more candidates than Swiss NSF. The biggest difference between the two approaches with respect to the outcome, is that the Swiss NSF rule assigns the same probability to every proposal that enters its lottery tier, whereas \ouralgorithm\ allows a range of probabilities: the broadest span is on NeurIPS Gaussian where marginal probabilities range from 0.16 to 0.83. This reflects the additional flexibility of \ouralgorithm\'s lottery. Finally, the relative weight \ouralgorithm\ puts in the ``certain accept'' group compared to Swiss NSF ($\prob=1$) is not uniform---it is higher than Swiss NSF for NeurIPS Subjectivity, essentially tied on the original Swiss NSF set, and lower for the two ICLR variants, so the split between guaranteed and lottery funding depends on the structure of the intervals in each study rather than following a single trend.

\begin{table}[ht]
  \centering
  \begin{tabular}{l|cccccc}
    \toprule
     & \multicolumn{3}{c}{\textbf{\ouralgorithm}} & \multicolumn{3}{c}{\textbf{Swiss NSF Algorithm}} \\
    \cmidrule(lr){2-4}\cmidrule(lr){5-7}
     Dataset & \% Accept & \% Random & Range of $\prob$ & \% Accept & \% Random & Value of $\prob$ \\
     \midrule
    Swiss NSF & 28.3 & 3.4 & (0.5, 0.5) & 28.0 & 3.7 & 0.54 \\
    Swiss NSF (Manski Bounds) & 5.4 & 26.9 & (0.75, 0.94) & 5.4 & 26.9 & 0.92 \\
    \midrule
    NeurIPS LOO & 3.4 & 16.5 & (0.36, 0.94) & 3.4 & 17.9 & 0.34 \\
    NeurIPS Gaussian & 2.2 & 25.7 & (0.16, 0.83) & 4.0 & 27.5 & 0.20 \\
    NeurIPS Subjectivity & 4.5 & 18.7 & (0.14, 0.45) & 1.6 & 10.4 & 0.78 \\
    \midrule
    ICLR LOO & 11.1 & 32.4 & (0.51, 0.88) & 22.0 & 22.8 & 0.44 \\
    ICLR Gaussian & 9.5 & 34.4 & (0.45, 0.87) & 21.2 & 27.7 & 0.4 \\
    ICLR Subjectivity & 17.6 & 24.3 & (0.50, 0.88) & 16.2 & 25.2 & 0.63 \\
    \bottomrule
      \end{tabular}
    \caption{Comparison of marginal probabilities of acceptance by \ouralgorithm\ and Swiss NSF on each dataset. ``Accept'' = guaranteed to be selected ($\prob = 1$), while ``Random'' = entered into lottery ($0 < \prob < 1$).}
      \label{tab:probability_results}
\end{table}

\section{Related Work}
\label{sec:related_work}

We now discuss other related works in scientific peer review, improving evaluation processes,  and general work on robust optimization.

\paragraph{Designing Better Evaluation Processes}
Many works diagnose flaws in peer review---such as miscalibration and arbitrariness in opinions of reviewers---and propose improvements~\cite{shah2022surveyextended}. Many models assume linear miscalibration~\cite{flach2010kdd,baba2013quality,roos2011calibrate,roos2012statistical,schichl2019consensual,cortes2021inconsistency}, including those used in the Swiss NSF~\cite{heyard2022rethinking}, but these often perform poorly in practice~\cite{langford2012icml}, likely due to the complexity of real-world miscalibration~\cite{brenner2005modeling}. A recent approach~\cite{wang2018your} addresses arbitrary miscalibration and shows that randomized estimators can improve ranking accuracy. Similarly to this, we use randomness to achieve robustness,but we model uncertainty more generally, via intervals, and focus on how to make decisions from these intervals. Related work in hiring proposes algorithmic solutions to prejudice and to uncertainty about unknown options~\cite{kleinberg2018selection,li2020hiring}. Both approaches make stronger assumptions on a known model of errors in hiring process than our work.

\paragraph{Selection Under Uncertainty and Robust Optimization}

Our method falls under the umbrella of distributionally robust optimization (DRO)~\cite{Rahimian_2022}. Our main contribution is a formulation specific to scenarios where decision makers want to select $\nselected$ high quality items under ambiguity sets defined by intervals. We design a novel, efficient optimization algorithm for this problem. Though similar in spirit to robust portfolio optimization~\cite{xidonas2020robust}, our assumptions differ: investors typically have probabilistic models and seek risk-averse diversification, while our setting assumes no reliable estimates of means or correlations. Recent work on conformal prediction for selection~\cite{jincandes2023} focuses on bounding false positives above a fixed threshold, whereas we aim to select the top-$\nselected$ items without assuming a fixed quality threshold or a means of estimating probabilities.

\paragraph{Randomized Social Choice}

Our work connects to the literature on randomized social choice, which examines how incorporating randomness into voting rules can enhance desired properties like fairness, consistency, and strategyproofness~\cite{brandt2016handbook,brandl2016consistent}. Unlike our setting, most social choice models do not assume an underlying ground-truth quality of items. However, one related line of work in social choice theory interprets voters’ rankings as noisy reflections of latent utilities and seeks aggregation rules that perform well under this uncertainty. In particular, ``distortion'' measures the worst-case loss in aggregate utility between the chosen alternative and the optimal one. Randomized mechanisms are known to achieve strictly better worst-case distortion guarantees than any deterministic rule~\cite{anshelevich2021distortion}. Our approach is conceptually similar: we consider all rankings consistent with uncertain intervals of quality and design a randomized mechanism that maximizes the worst-case expected utility of the selected top-$\nselected$ items.

\section{Discussion}
\label{sec:discussion}

We introduce \ouralgorithm, a computationally efficient and principled lottery for top-$\nselected$ selection under uncertainty.  By relying solely on interval information, rather than a fully specified generative model, \ouralgorithm\ respects the limits faced by real-world funders while providing a justification for use as the solution to a robust optimization problem. Our case studies show that \ouralgorithm\ scales well to tens of thousands of candidates. 

An additional benefit of \ouralgorithm\ is that it can easily be adapted to handle constraints on the form of the lottery. For example, in some cases, funders may prefer to implement a \emph{uniform lottery}, where all candidates subject to randomization are selected with equal probabilities---this simple form of lottery may be more acceptable and explainable to a broad audience. The optimization problem used in \ouralgorithm\ can be modified to constrain the lottery to uniform sampling, thereby implementing the ex ante optimal uniform lottery that respects ex post validity constraints. We discuss this extension in Appendix~\ref{app:uniform_lottery} and show through simulations that it performs similarly to other methods in terms of expected utility, while providing significant worst-case utility gain under our model.

\paragraph{Limitations.}  
We now describe a number of limitations to our approach:
\begin{itemize}
    \item \ouralgorithm\ assumes the decision maker exploits only the \emph{ordering} of intervals.  This is appropriate when credible probabilistic models are unavailable, but may be sub-optimal in domains where well-calibrated predictive models exist. If a decision-maker has access to a trustworthy probabilistic model, they may choose deterministically based off of point estimates of their model, whereas our approach takes a worst-case applicant.
    \item \ouralgorithm\ currently assumes that all proposals request the same budget. In practice, some funding scenarios allow applicants to request variable amounts of funds. This requires the funder to account both for the relative quality of proposals and for the variable amount of the total budget used up by each proposal. We propose a heuristic solution for this in future work. 
    \item Decision makers may possess additional information about proposal quality that is easily captured in our framework. For example, a funder may have additional information about the relative quality of a pair of proposals with overlapping intervals, if the same reviewer evaluated both and concluded one proposal was higher quality than the other. 
    \item \ouralgorithm\ may be less interpretable than Swiss NSF and randomize-above-threshold, which both set a single threshold for randomization.

\end{itemize}

\paragraph{Future work.}  
Extending \ouralgorithm\ to more general settings and understanding additional reasons to randomize raise several interesting research questions:

\begin{itemize}[leftmargin=*]

    \item \emph{Variable-cost candidates.}  
    Grants often request \emph{variable} amounts of funding subject that a funder wishes to allocate subject to an overall budget constraint. A natural extension of the \ouralgorithm\ approach is to allow partial funding of a candidates' funding request. Then, when a candidate requests (integral) $c$ resources, the algorithm inputs the candidate into the algorithm as $c$ identical copies of their interval. This approach raises incentive issues around truthful cost reporting.  Further, if candidates' funding requests cannot be partially funded, then it is unclear how to extend \ouralgorithm.

    \item \emph{Additional information on relative order of proposals.} Decision-makers may obtain additional information that allow them to compare some proposals beyond the interval order. Solving the ex ante optimization problem for a general partial order is NP hard, so it is worth considering the performance of heuristics for solving this problem when the decision-maker starts with an interval order and adds a small number of additional comparisons. They could, for example, impose monotonicity constraints on $\prob$ for pairs of proposals for which they know additional ordering information.

    \item \emph{Richer utility functions.}  
    Our problem formulation assumes a 0-1 utility when selecting the top-$\nselected$ candidates out of a ranking of candidates. A decision-maker may have a richer utility function (for example, using Borda count or any positional scoring rule for the selected ranks of each proposal.)  Understanding for which utility functions a \ouralgorithm\-style maximin solution remains tractable is an open theoretical and algorithmic challenge.  

    \item \emph{Cost–quality trade-offs.}  
    A common argument for lotteries is reviewer-time savings. An analysis of efficient trade-offs between reviewer resources and decision quality may yield new insights around when and how to randomize.

    \item \emph{Equilibrium and behavioral effects.}  
    Randomization may alter incentives for applicants and reviewers. Such equilibrium and behavioral effects should be accounted for in a complete analysis of the usefulness of randomization. For example, applicants may put more effort into their applications if they are risk-averse, or may be more willing to submit riskier proposals if they know there is a chance of acceptance due to lottery. On the other hand, reviewers may engage in new forms of strategic behavior in order to alter the outcomes of the lottery. 
\end{itemize}

Progress on these questions can help inform both theoretical understanding and practical implementation of randomized selection mechanisms.

\section*{Acknowledgments}
This work was supported in parts by NSF 1942124 and ONR N000142212181, as well as the Gates Foundation.

\bibliographystyle{alpha}
\bibliography{bibtex}
~\\~\\~\\
\appendix

\section*{Appendix}

\section{Proofs}
\label{app:proofs}

\subsection{Optimality of deterministic selection in a fully Bayesian setting (Section~\ref{sec:existing_deployments})}

\label{app:bayesian_deterministic}

\begin{proof}[Proof of Proposition~\ref{prop:fully_bayesian_model}]
    Expand the expected utility by conditioning on the observed data $y$ as $\E_{y, \theta, \hat{\pi}} [ u(\hat{\pi}, \theta)] = \E_{y}[\E_{\hat{\pi}, \theta | y}[u(\hat{\pi}, \theta)]]$. For a given $y$, the funder maximizes the inner expectation by choosing any $\hat{\pi} \in \argmax_{\pi' \in \Pi} \E_{\theta | y}[ u(\pi', \theta) | y]$, where the maximum over $\Pi$ always exists since the set of permutations $\Pi$ is finite. Hence, the deterministic estimator that takes $f(y) \in \argmax_{\pi' \in \Pi} \E_{\theta | y}[ u(\pi', \theta) | y]$ for all $y \in Y$ maximizes the funder's expected utility.
\end{proof}

\subsection{Polynomial time algorithm for ex ante optimization~(Theorem~\ref{prop:ellipsoid})}

In order to prove the main theorem, we first prove Lemma~\ref{proposition:separation_oracle}, which establishes that the separation oracle (Algorithm~\ref{alg:separationoracle}) is correct and efficient.

\begin{proof}[Proof of Lemma~~\ref{proposition:separation_oracle}]
    Let $\objvalue'$ be the smallest $\objvalue_i$ found by the separation oracle. We begin by showing that:
    \[\objvalue' = \min_{\perm \in \errorpermutations} \; \sum_{i = 1}^{\nproposal} \prob_i \ind\{\perm(i) \leq \nselected\}.\]
    For $i \in [\nproposal]$, let 
    \[K = \{T \subset[\nproposal], |T| = \nselected: \exists \perm \in \errorpermutations \text{ s.t. } \perm(j) \leq \nselected \; \forall j \in S \}\] be the set of all feasible selections of top-$\nselected$ intervals given $\errorpermutations$. Note that we can rewrite the objective as
    \[\min_{\perm \in \errorpermutations} \; \sum_{i = 1}^{\nproposal} \prob_i \ind\{\perm(i) \leq \nselected\} = \min_{T \in K} \sum_{j \in T} \prob_j.\]  
    Now, we divide $K$ into subsets, for $i \in [\nselected+1]$, define 
    \[K^{(i)} = \{T \subseteq K: i \not \in T \text{ and } [i-1] \subset  T\}\]
    so that $K^{(i)}$ contains all sets in $K$ that include intervals $1$ to $(i-1)$ and exclude interval $i$ and $K = \cup_{i=1}^{\nselected+1} K^{(i)}$. Let the intervals be sorted in decreasing order of lower bound $\lcb$. Then, for any $i \in [\nproposal]$, $i$ is not strictly above any of the intervals from $1$ to $(i-1)$. Additionally, letting 
    \[S_i = \{j \in (i, \nproposal]: \text{intervals } j \text{ and } i \text{ overlap} \}\]
    any interval in $(i, \nproposal] \setminus S_i$ must be strictly below interval $i$. Hence, any set of top-$\nselected$ intervals $T \in K^{(i)}$ contains intervals $1$ to $(i-1)$ and the remaining $\nselected - (i-1)$ intervals come from $S_i$. Additionally, for all $j \in S_i$, $\ucb_j \geq \lcb_i$, while $\lcb_j \leq \lcb_i$ so all intervals in $S_i$ overlap each other and any selection of these $S_i$ in the top $\nselected$ is feasible. Hence:
      \[K^{(i)} = \{[i-1] \cup A \; : \;  A \subseteq S_i, |A| = (\nselected - (i-1)\}.\]
      
      Finally, for fixed $\prob$, the minimum objective value over $K^{(i)}$ is to select the $(\nselected - (i-1)$ proposals with the smallest values of $\prob$. Algorithm~\ref{alg:separationoracle} exactly enumerates these worst-case constraints for each $K^{(i)}$ and hence returns the minimum over $T \in K$.

    Now, the separation oracle adds constraints to $\cuts$ when $ \objvalue > \sum_{j=1}^{i-1} \prob_j + \sum_{j \in S_i} \prob_j$. Hence the separation oracle returns $\emptyset$ only if $\objvalue \leq \sum_{j=1}^{i-1} \prob_j + \sum_{j \in S_i} \prob_j$ for every $i$, which is precisely the condition for feasibility. Further, as argued above, the separation oracle considers feasible top-$\nselected$, so any constraints added by the separation oracle represent valid constraints for the LP.
    
      The algorithm runs in time $O(\nproposal \max\{\nselected, \log \nproposal\})$, because finding $S_i$ requires at most a linear scan over $O(\nproposal)$ proposals for each of $(\nselected)$ proposals and initial sorting takes time $O(\nproposal \log \nproposal)$. We can sort intervals by $\prob$ once, to ensure that each linear scan for $S_i$ returns the intervals with smallest $\prob$ with no additional sorting per iteration.
    \end{proof}

Now, we prove the main theorem by showing that the ellipsoid algorithm converges in polynomial time using this separation oracle as a sub-routine:

\begin{proof}[Proof of Theorem~\ref{prop:ellipsoid}]
    We may relax the equality constraint of LP~\eqref{LP} $\sum_{i=1}^{\nproposal} \prob_i = \nselected$ to an inequality $\sum_{i=1}^{\nproposal} \prob_i \leq \nselected$ without loss of optimality, since increasing any $\prob_i$ (up to 1) cannot decrease the worst-case value of $\objvalue$. This relaxation ensures that the feasible region lies within a full-dimensional affine subspace of $\mathbb{R}^{\nproposal + 1}$, and contains a nontrivial interior. In particular, the point $\prob_i = \frac{\nselected}{\nproposal}, \objvalue = 0$ lies strictly inside the box constraints and satisfies all inequalities strictly, implying the feasible region is full-dimensional.
    
    Then, because our separation oracle runs in time polynomial in $\nproposal$, the classical ellipsoid algorithm~\cite{khachiyan1979polynomial} using the separation oracle to make cuts, solves the optimization problem to within accuracy $\epsilon$ in time $\text{poly}(n, \log(1/\epsilon), \log(R/r))$ where $R$ is the radius of a Euclidean ball that contains the feasible region and $r$ is the radius of a Euclidean ball entirely contained in the feasible region. Clearly, $R$ is upper bounded by $\text{poly}(\nproposal)$, since each $\prob_i$ is bounded in $[0,1]$ and $\objvalue$ is bounded in $[0,\nselected]$. To establish a lower bound on $r$, we invoke Theorem 6.2.2 in Grötschel, Lovász, and Schrijver \cite{grotschel1988geometric}, which states that if a polyhedron $P = {x \in \mathbb{R}^n : Ax \leq b}$ is full-dimensional, and the matrix $A$ and vector $b$ consist of integers of maximum bit length $U$, then $P$ contains a ball of radius at least $2^{-\text{poly}(n, U)}$. Noting that all constraints in our model have integral coefficients of bit length at most $\log(\nselected) < n$, we have that $r \geq 2^{-\text{poly}(n)}$. Hence, $\log(R/r)$ is $\text{poly}(\nproposal)$ and so the runtime of the algorithm is polynomial in $\nproposal$ and $\log(1/\epsilon)$.
    \end{proof}

\subsection{Ex post validity~(Theorem~\ref{thm:ex_post})}

\begin{proof}[Proof of Theorem~\ref{thm:ex_post}]
    We will prove that for any $a, b \in [\nproposal]$ such that $\lcb_a > \ucb_b$ and for any input vector of marginal probabilities $\prob$, the post-processing Algorithm~\ref{alg:tie_breaking} terminates with $\prob_a = 1$ or $\prob_b = 0$ and never decreases the objective value of $\prob$. Then, if any ex ante optimal $\prob$ is given as input to the algorithm, after post-processing it is still ex ante optimal. Moreover, any sampling method that respects the post-processed marginal probabilities will satisfy ex post validity, since if $a$ dominates $b$, either $a$ is always sampled $(\prob_a = 1)$ or $b$ is never sampled $(\prob_b = 0)$.  
    
    Let $a, b$ be any intervals with $\lcb_a > \ucb_b$. Let $D = \{d \in [\nproposal] \mid \lcb_d > \ucb_a\}$ be the set of all intervals strictly above interval $a$. Note that (1) all $d \in D$ are also strictly above $b$ and that (2) $b$ precedes $a$ when ordered by $\ucb$ and $a$ precedes all $d \in D$. Hence, Algorithm~\ref{alg:tie_breaking} will process $b$, then $a$, then $d \in D$. After processing $b = b$, either $\prob_b = 0$ or $\prob_d = 1$ for all $d \in D$. Because $\prob_b$ cannot increase in any subsequent iterations, if $\prob_b = 0$, it will remain $0$ until the algorithm terminates. If $\prob_d = 1$ for all $d \in D$, then because $a \in D$, $\prob_a = 1$. Furthermore, since $\prob_d = 1$ for all $d \in D$, $\prob_a$ will not decrease in any subsequent iterations. Hence, the algorithm ends with $\prob_b = 0$ or $\prob_a = 1$.

    Further, for any $\perm \in \errorpermutations$, if $\perm(b) \leq \nselected$, then $\perm(a) < \perm(b) \leq \nselected$. Hence, moving probability mass from $\prob_b$ to $\prob_a$ cannot decrease $\min_{\perm \in \errorpermutations} \; \sum_{i = 1}^{\nproposal} \prob_i \ind\{\perm(i) \leq \nselected\}$.
\end{proof}

\subsection{Axiomatic comparison~(Theorem~\ref{thm:axiomatic_analysis})}
\label{ref:axiom_proofs}

\paragraph{(a) Maximum instability}

We will show that both Swiss NSF and randomize above threshold (with data-dependent threshold) are maximally unstable using an example. Let intervals $1$ to $\nselected$ be $[0,2]$ with point estimates of $1$ and let intervals $\nselected+1$ to $\nproposal$ be $[0, 1 - \epsilon]$ for some $\epsilon \in (0, 1)$. Then, the Swiss NSF algorithm selects the first $\nselected$ proposals deterministically. Now, shift interval $\nselected$ to be $[0,2-2\epsilon]$ with point estimate of $1-\epsilon$. Since all the intervals contain the $\nselected$-the point estimate the Swiss NSF algorithm selects uniformly at random. Note that $\epsilon$ can be taken to be arbitrarily small, so the Swiss NSF is maximally unstable for any $\epsilon > 0$. If the randomize above threshold method chooses the threshold in a data-dependent manner, then the same example leads to maximal instability if the threshold is taken to be the $\nselected$-th highest point estimate.  

Now, to show that \ouralgorithm\ is not maximally unstable, we first observe that \ouralgorithm\ selects among all proposals uniformly at random \emph{only if} all intervals intersect each other. Let $1 < \nselected < \nproposal-1$. Assume for the sake of contradiction that \ouralgorithm\ samples uniformly at random from all intervals and that there are two intervals $i$ and $j$ with $i$ strictly above $j$. Since the algorithm is ex post valid, if $\prob_j > 0$, then $\prob_i = 1$, but then $\prob_i = \prob_j = 1$, which is possible only if $\nproposal = \nselected$, yielding a contradiction.

Now, to show that \ouralgorithm\ is \emph{not} maximally unstable, we establish that for any set of intervals where all intervals except for one overlap, the algorithm never chooses deterministically. Since all intervals but one overlap, the non-overlapping interval can either have the largest lower bound or the smallest upper bound.

First, consider the case where interval $j$ has the largest lower bound. Then, we can partition the intervals into $3$ sets $\{j\}$, $X$, the set of intervals that intersect $j$, and $B$ the set of intervals strictly below $j$. By symmetry and monotonic chain constraints (detailed in Appendix~\ref{app:full_cutting_plane}), the algorithm will output marginal probabilities with at most $3$ different values $\prob_j, \prob_x, \prob_b$ such that $\prob_j \geq \prob_x \geq \prob_b$. Since $|X| + |B| > \nselected$, the only input for which the algorithm could output $\nselected$ $1$'s is when $|X| = \nselected - 1$ and $|B| = n-\nselected > 1$. In this case, $j$ is always in the top $\nselected$, so $\prob_j = 1$ and $j$ is pruned from the problem. Then, our algorithm chooses $\nselected - 1$ out of $\nproposal - 1$ intervals that all intersect, so the optimal solution is sampling $\nselected - 1$ uniformly at random from the remaining $X \cup B$ intervals.  Hence, our algorithm will not choose deterministically.

Now, consider the case where interval $j$ has the smallest upper bound. Then, we can partition the intervals into $3$ sets $\{j\}$, $X$, the set of intervals that intersect $j$, and $A$, the set of intervals strictly above $j$. By symmetry and monotonic chain constraints, our algorithm will output marginal probabilities with at most $3$ different values $\prob_j, \prob_x, \prob_a$ such that $\prob_a \geq \prob_x \geq \prob_j$. Since $|A| + |X| \geq \nselected + 1$, the only possible input for which the algorithm could output $\nselected$ $1$' s is if $|A| = \nselected$ and $|X| = \nproposal - \nselected - 1 > 0$. In this case, $j$ is never in the top $\nselected$, so the feasible top $\nselected$ could be any subset of size $\nselected$ from $A \cup X$. Therefore, selecting $\nselected$ proposals uniformly at random from $A \cup X$ is optimal, so \ouralgorithm\ will not choose deterministically. \qed

\paragraph{(b) + (c) Monotonicity in budget}

We show that Swiss NSF, randomize above threshold, and \ouralgorithm\ are not monotonic in budget $\nselected$, via the example shown in Figure~\ref{fig:monotonicity_failure_app}.

\begin{figure}[ht]
    \centering
    \includegraphics[width=0.4\linewidth]{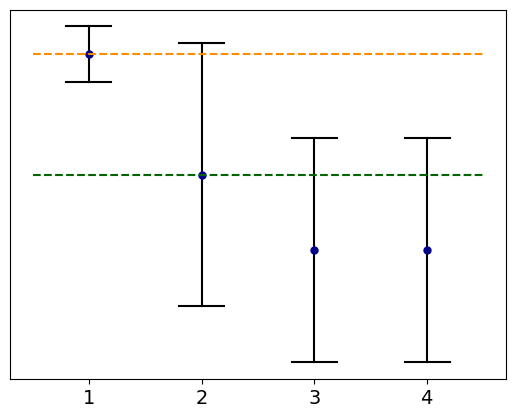}
    \caption{Example that violates monotonicity with respect to $\nselected$ for Swiss NSF and our \ouralgorithm\ algorithm. When $\nselected=1$, $\prob_2 = 1/2$ for both algorithms. However, when $\nselected=2$, $\prob_2 = 1/3$ for both algorithms.}
    \label{fig:monotonicity_failure_app}
\end{figure}

In this example, taking $\nselected = 1$, both the Swiss NSF method and \ouralgorithm\ randomize between proposals $1$ and $2$, so proposal $2$ is selected with probability $1/2$. However, if the number of proposals selected is increased to $\nselected=2$, then proposal $1$ is always selected and both algorithms sample uniformly at random from intervals $2$ to $4$, meaning that proposal $2$ has a selection probability of only $1/3$. Hence, even though more proposals are being selected, proposal $2$ is worse-off. The same example shows that randomize-above-threshold violates monotonicity, taking the threshold to be the point estimate or the lower bound of the $\nselected$-th highest proposal. 

This example also proves that it is not possible to simultaneously satisfy ex ante optimality and monotonicity in budget. For $\nselected=1$, the unique ex ante optimal $\prob$ is $(1/2, 1/2, 0, 0)$ since feasible top $1$ are $\{1\}$ and $\{2\}$. For $\nselected=2$, the unique ex ante optimal $\prob$ is $(1, 1/3, 1/3, 1/3)$ since the feasible top $2$ are $\{1,2\}$ $\{1,3\}$, $\{1,4\}$. Since these are unique ex ante optimal solutions, there is no sequence of solutions for $\nselected$ equal to $1$ and $2$ that satisfy monotonicity in budget $\nselected$ and ex ante optimality simultaneously.

\paragraph{(d) Reversal symmetry}

To prove that Swiss NSF and (data-dependent) randomize-above-threshold violate reversal symmetry, consider intervals $(0,1)$ and $(0.1, 0.2)$ with point estimates $0.5$ and $0.15$. Then, the Swiss NSF selection rule or randomize-above-threshold with threshold set as the highest point estimate, will accept interval $1$ and reject interval $2$. However, if intervals are flipped so interval $1$ stays the same, but proposal $2$ has interval $(0.8, 0.9)$ with point estimate of $0.85$, then Swiss NSF samples uniformly at random between the two proposals. Now, when $\nproposal=2$ and $\nselected = 1$, \ouralgorithm\ only samples between the two proposals if their intervals overlap. If the intervals are horizontally flipped, then they still overlap and \ouralgorithm\ samples uniformly at random respecting reversal symmetry. If interval $1$ lies strictly above interval $2$, then after flipping interval $2$ lies strictly above interval $1$, so the marginal probabilities of selection are $(1,0)$ and $(0,1)$ respectively, respecting reversal symmetry. 

\subsection{Enforcing monotonicity (Proposition~\ref{prop:monotonicity_algo})}

\begin{proof}
    Let $\prob^{(i-1)}$ be the output at step $i-1$ and $\prob'^{(i)}$ be the output for iteration $i$ before post-processing and $\prob^{(i)}$ be the final output after post-processing. If an interval $b$ is strictly below a set of intervals $A$, either $\prob^{(i-1)}_b = 0$ or $\prob^{(i-1)}_b > 0$ and $\prob^{(i-1)}_a > 0  \; \forall a \in A$. If $\prob^{(i-1)}_b = 0$, clearly $\prob^{(i)}_b >= \prob^{(i-1)}_b$. If $\prob^{(i-1)}_b > 0$ and $\prob^{(i-1)}_a > 0  \; \forall a \in A$, then $\prob'^{(i)}_a = 1$ for all $a \in A$ and $\prob'^{(i)}_b \geq \prob^{(i-1)}_b$, so the post-processing will not decrease $\prob'^{(i)}_b$ at all and monotony will be satisfied.
\end{proof}

\section{Full Cutting Plane Algorithm}
\label{app:full_cutting_plane}

In Section~\ref{sec:practical_algorithm}, we describe the cutting plane algorithm at a high-level. This section proves that pruning and adding monotonicity and symmetry constraints to the LP are without loss of optimality and ensure faster convergence. 

We note that one can think of $\nabove, \nbelow$ as defining a partial order: $i \succeq j \iff \nabove(i) \leq \nabove(j) \text{ and } \nbelow(i) \geq \nbelow(j)$. Then, a monotonically ordered subset of intervals is a totally ordered subset or a ``\emph{chain}'' and the minimum number of monotonically ordered subsets of intervals that covers a set of intervals is the ``\emph{width}'' of the partially ordered set of intervals. Note that this partial order defined by $\nabove$ and $\nbelow$ is \emph{not} the canonical interval order that we use to define the set of feasible permutations $\errorpermutations$. The ($\nabove, \nbelow$) partial order never has width larger than the interval order and often has much smaller width. At the extreme, if all intervals overlap, then the width of the interval partial order is $\nproposal$, while the width of the ($\nabove, \nbelow$) partial order is $1$, since all proposals have equal $\nabove$ and $\nbelow$. 

As we now prove, the cutting plane algorithm converges in $O(\nselected^{\orderwidth+1})$ iterations, where $\orderwidth$ denotes the number of monotonically ordered subsets partitioning the intervals (or the width of the ($\nabove$, $\nbelow$) partial order. In theory, $\orderwidth$ can grow linearly with $\nselected$, so this algorithm does not have the theoretical polynomial-time guarantee of the ellipsoid algorithm. In practice, $\orderwidth$ is often small for many sets of intervals.

\begin{proposition}[Cutting plane algorithm convergence]
\label{proposition:cutting_plane_algo}
    Letting $\orderwidth$ denote the number of monotonically ordered subsets (chains) partitioning the intervals per Definition~\ref{defn:monotonically_ordered_subsets}, Algorithm~\ref{alg:cutting_plane} converges to an optimal solution in $O(\nselected^{\orderwidth+1})$ iterations. The algorithm solves an LP with at most $O(\nproposal + \nselected^{\orderwidth+1})$ constraints.
\end{proposition}

\begin{proof}

We first prove three lemmas that show that we can impose initial pruning, symmetry, and monotonicity constraints on $\prob$ without loss of optimality.

\begin{lemma}[Pruning of optimal $\prob$]
\label{proposition:pruning}
There exists an optimal $\prob$ in which $\prob_i = 1, \forall i \in [\nproposal]$ with $\nbelow(i) \geq \nproposal - \nselected$ and $\prob_j = 0, \forall j \in [\nproposal]$ with $\nabove(j) \geq \nselected$.     
\end{lemma}

\begin{proof}
    If $\nbelow(i) \geq \nproposal - \nselected$, then $i$ is always included in the top $\nselected$ in any permutation of intervals. If $\prob_i < 1$, then there must be some other interval with $\prob_j > 0$, but setting $\prob_i = 1$, $\prob_j = 0$ will not decrease the objective value, since $i$ is always in the top $\nselected$. Similarly, if $\nabove(i) \geq \nselected$, then  $i$ is never included in the top $\nselected$, so taking $\prob_i = 0$ by shifting probability mass from $i$ to any other proposal cannot hurt the objective value.
\end{proof}

\begin{lemma}[Symmetry of optimal $\prob$] 
\label{proposition:symmetry}
There exists an optimal $\prob$ in which $\prob_i = \prob_j$ for all $i,j \in [\nproposal]$ such that $\nabove(i) = \nabove(j)$ and $\nbelow(i) = \nbelow(j)$.
\end{lemma}

\begin{proof}
    If $\nabove(i) = \nabove(j)$ and $\nbelow(i) = \nbelow(j)$, then $i$ and $j$ have the same sets of intervals that are strictly above and strictly below each interval. Hence, for any permutation $\perm \in \errorpermutations$ the permutation $\perm'$ with $i$ and $j$ exchanged is also in $\errorpermutations$, so the objective value is maximized at $\prob_i = \prob_j$.
\end{proof}

\begin{lemma}[Ordering of optimal $\prob$ by $\nabove$ and $\nbelow$]
    \label{proposition:ordering}
  Let $M_1, \ldots, M_\orderwidth$ be a partitioning of $[\nproposal]$ such that each $M_i$ is monotonically ordered. Then, there exists an optimal $\prob$ for Objective~(\ref{ref:opt_problem}) such that within each $M$, $p_{M[j]} \geq p_{M[j+1]}$ for all $j \in [|M| - 1]$.
\end{lemma}

\begin{proof}
    Let $i,j \in [\nproposal]$ be any pair of proposals with $\nabove(i) \leq \nabove(j)$ and $\nbelow(i) \geq \nbelow(j)$. Let $p$ be a feasible solution to the optimization problem with $p_j > p_i$. We will show that exchanging the values of $p_j$ and $p_i$ can never decrease the objective value.
    Define $q$ as equivalent to $p$ with $i$ and $j$ exchanged: \[
        q_r = \begin{cases}
            p_j & r = i \\ 
            p_i & r = j \\ 
            p_r & \text{otherwise}.
        \end{cases}
    \]
    Let $\objvalue(p, \perm) = \sum_{i = 1}^{\nproposal} \prob_i \ind\{\perm(i) \leq \nselected\}$. We want to show that $\min_{\perm \in \errorpermutations} v(p, \perm) \leq \min_{\perm \in \errorpermutations} v(q, \perm)$. 
    Consider any $\perm \in \errorpermutations$. If $\perm(i) < \perm(j)$, then $\objvalue(p, \perm) \leq \objvalue(q, \perm)$. If $\perm(i) > \perm(j)$, define permutation $\tau$ to be equivalent to $\perm$ but with $i$ and $j$ exchanged:
    \[
        \tau(r) = \begin{cases}
            \perm(j) & r = i \\ 
            \perm(i) & r = j \\ 
            \perm(r) & \text{otherwise}.
        \end{cases}
    \]
    Note that $\objvalue(p, \perm) \geq \objvalue(q, \perm)$, but $\objvalue(q, \tau) = \objvalue(p, \perm) \geq \objvalue(p, \tau)$. We now proposition that $\tau \in \errorpermutations$. Let $r \in [\nproposal]$ be any proposal such that $\perm(r) \in [\perm(j), \perm(i)]$. Then, $r$ cannot be strictly above $j$, but since $\nabove(j) \geq \nabove(i)$, $r$ cannot be strictly above $i$. Similarly, $r$ cannot be strictly below $i$, but since $\nbelow(i) \geq \nbelow(j)$, $r$ cannot be strictly below $j$. Hence, $r$ must overlap both $i$ and $j$. Therefore, exchanging $i$ and $j$ does not violate any constraints, so $\tau \in \errorpermutations$. Hence, we conclude that 
    $\min_{\perm \in \errorpermutations} \objvalue(p, \perm) \leq \min_{\perm \in \errorpermutations} \objvalue(q, \perm)$. Now, applying this exchange to every sequence of intervals within each $M_i$ yields the desired result.
\end{proof}

Further, using symmetry and pruning together reduces the number of decision variables to $O(\nselected)$ from $O(\nproposal)$.

\begin{lemma}[Pruning and symmetry give $O(\nselected)$ decision variables]
After pruning per Proposition~\ref{proposition:pruning} and grouping symmetric intervals per Proposition~\ref{proposition:symmetry}, the optimization problem has $< 2\nselected$ decision variables.
\label{proposition:num_decision_vars}
\end{lemma}

\begin{proof}
    After pruning all proposals with $\nabove(i) \geq \nselected$ or $\nbelow(i) \geq \nproposal - \nselected$, there are at most $\nselected - 1$ proposals with $\nbelow(i) > 0$. The remaining intervals all have $\nbelow(i) = 0$ and can have $\nselected - 1$ possible values of $\nabove(i)$, so they can be grouped into at most $\nselected - 1$ groups with equivalent $\prob_i$. Hence, the intervals can be grouped into at most $2\nselected - 2$ decision variables. 
\end{proof}

    Now, to complete the proof of correctness, observe that by Lemmas~\ref{proposition:pruning} and \ref{proposition:ordering} the initial pruning and monotonicity constraints are without loss of optimality. We can additionally impose equality constraints on $\prob$ using symmetry (Lemma~\ref{proposition:symmetry}) to reduce the dimension of $\prob$ to $O(\nselected)$, although this is not shown in Algorithm~\ref{alg:cutting_plane} for simpler presentation. Because the linear program solved by the cutting plane algorithm is a relaxation of the linear program~(\ref{LP}), a solution to the problem upper bounds the objective value of (\ref{LP}). By the correctness of the separation oracle (\Cref{proposition:separation_oracle}), if the cutting plane algorithm converges to a feasible $(\prob, \objvalue)$, this therefore is an optimal solution to the full LP.

    To prove that the algorithm is guaranteed to converge within $\nselected^{\orderwidth+1}$ iterations observe that there are at most $O(\nselected^\orderwidth)$ possible total orders of marginal probabilities (after applying symmetry and pruning) consistent with the partial order given by monotonicity constraints $\prob_{M[i]} \geq \prob_{M[i+1]} \quad \forall i \in [|M|-1], \forall M \in \{M_1, \ldots, M_\orderwidth$. Note that for any total order of the proposals, there are $\nselected+1$ possible constraints that the separation oracle can return, because the possible constraints are determined by the order of the $\prob_i$. Hence, there are at most $O(\nselected^{\orderwidth+1})$ cuts that could be added to the linear program and so the cutting plane algorithm must converge in $O(\selected^{\orderwidth+1})$ iterations. The initial LP only contains $O(\nproposal)$ constraints, so the LP never has more than $O(\nproposal + \nselected^{\orderwidth+1})$ constraints.
\end{proof}

\subsection{Finding Minimal Set of Chains for  \texorpdfstring{$(\nabove, \nbelow)$}{(\nabove, \nbelow)} Partial Order}
\label{app:chain_partition}

In order to partition intervals using $\nabove$ and $\nbelow$ as per 
Proposition~\ref{proposition:ordering}, we need to compute such a partition, known as a chain covering of the set. We would like the partition into as few sets as possible in order to add as many constraints as possible to the problem and reduce runtime of our optimization algorithm. There are practical general methods to solve this chain cover problem for any partial order in time $O(\nproposal^{2.5})$ by computing the maximum matching of an appropriately constructed bipartite graph. In our case, where the partial order has specific structure, we can solve the minimal chain cover problem even more efficiently in time $O(\nproposal \log \nproposal)$. The algorithm described below is equivalent to an algorithm given in \cite{golumbic1980algorithmic}[Chapter 7, Algorithm 7.1] for minimal coloring of a permutation graph. For completeness, we reproduce the algorithm and proof of optimality in our problem setting below. 

\begin{algorithm}[H]
\caption{Greedy minimal chain cover for the product order on $\mathbb{R}^2$}
\label{alg:chain-cover}
\begin{algorithmic}[1]
\Require A set of points $\{(a_i,b_i)\}_{i=1}^{n}$ with partial order $i \succeq j \iff  a_i \geq a_j \text{ and } b_i \geq b_j$.
\Ensure A set $\mathcal{C}=\{C_1,\dots,C_v\}$ of non-increasing chains that covers $[n]$.
\State Sort the indices in decreasing order of $a$ so that
      $a_{i_1}\ge a_{i_2}\ge\ldots\ge a_{i_n}$, breaking ties by decreasing $b$
\State $\mathcal{C}\gets\varnothing$
\For{$t\gets 1$ \textbf{to} $n$}
    \State Let $C$ be the chain with the smallest $b_{\text{tail}(C)}$ such that $b_{\text{tail}(C)} \ge b_{i_t}$
    \If{$C$ exists}
    \State Append $i_t$ to end of $C$
    \Else
    \State Create a new chain $C_{\text{new}}\gets\{i_t\}$ and append it to $\mathcal{C}$
    \EndIf
\EndFor
\State \Return $\mathcal{C}$
\end{algorithmic}
\end{algorithm}

\begin{proposition}[Correctness of minimal chain cover algorithm]
Algorithm~\ref{alg:chain-cover} returns a minimal size chain cover
$\mathcal{C}=\{C_1,\dots,C_v\}$ of the partially ordered set $([n],\succeq)$ defined by $
i\succeq j \Longleftrightarrow a_i\ge a_j\text{ and }b_i\ge b_j$. It runs in time $O(\nproposal \log \nproposal)$.
\end{proposition}

\begin{proof}
First, we proposition that each $C_r$ is a chain. Indices are processed in non-increasing $a$-order, so within any chain
the $a$-coordinates never increase.
A point $p$ is appended to $C_r$ only if the current tail $q$ of
$C_r$ satisfies $b_q\ge b_p$; hence
$(a_q,b_q)\succeq(a_p,b_p)$. 

Now, to show that this set of chains is minimal we will invoke Dilworth's Theorem, which states that the minimal number of chains that cover a poset (its width) is equivalent to the length of the longest antichain (sequence of incomparable elements.) We argue that the set of final tails of each chain forms an antichain. Let $T=\{t_1,\dots,t_v\}$ be the final tails ordered so that
$b_{t_1}<b_{t_2}<\dots<b_{t_k}$. For any chains $i,j$ with $i < j$ we have $a_{t_i} \geq a_{t_j}$, otherwise $t_j$ would have been added to chain $i$. But, $b_{t_i}<b_{t_j}$, so $t_i$ and $t_j$ are incomporable. Thus the elements of $T$ are pairwise
incomparable and $T$ is an antichain of size $v$. Now, let $A\subseteq[n]$ be any antichain. Map each $x\in A$ to the
chain that contains it.  Two distinct elements of $A$ cannot lie in the
same chain, hence $|A|\leq v$. Therefore, the length of the longest antichain is $v$, so the width of the partially ordered set is $v$ and the algorithm has returned a minimal number of chains.

The algorithm can be implemented in time $O(n \log n)$, because the set of chains is always sorted in increasing order of the $b$ value of their tails so we use binary search to find the chain in which to insert each $i$ (step $4$.)
\end{proof}

\section{Systematic Sampling}
\label{app:systematic_sampling}

The method known as ``Systematic Sampling'', works by first computing cumulative probabilities $S_i = \sum_{j=1}^i \prob_j$. Then, it selects a random starting point $u$ uniformly from the interval [0,1) and picks exactly one item from each of the $\nselected$ intervals obtained by adding integers $m=0,1,\dots,k-1$ to the starting point $u$. Each item is selected if at least one of these evenly spaced points falls within its corresponding cumulative interval $[S_{i-1}, S_i)$. Thus, the algorithm guarantees selecting exactly $\nselected$ distinct items without replacement, each with the correct marginal probability $\prob_i$. Additionally, we can initially shuffle the items uniformly at random, so that the algorithm replicates the expected behavior of uniform random sampling for items with equal values of $\prob$. The algorithm requires two passes over all the proposals and hence runs in time $O(\nproposal)$.

\begin{algorithm}[H]
\caption{Systematic Sampling~\cite{madow1949theory}}
\label{alg:systmetaic_sampling}
\begin{algorithmic}[1]
\Require Integer $\nselected$, probability vector $\prob \in [0,1]^\nproposal$ with $\sum \prob_i = k$
\Ensure A subset of $\nselected$ items sampled without replacement from $[\nproposal]$ where item $i \in [\nproposal]$ is included with marginal probability $\prob_i$.

\State Compute cumulative sums: $S_0 \gets 0$, $S_i \gets \sum_{j=1}^i \prob_j$ for $i = 1$ to $\nproposal$
\State Sample $u \sim \text{Uniform}(0, 1)$
\For{each $m$ from $0$ to $\nselected-1$}
\State Include item $i$ in the sample where $(u + m) \in [S_{i-1}, S_i)$
\EndFor
\end{algorithmic}
\end{algorithm}

\section{Additional Experiments Varying Model Parameters}
\label{app:additional_data}

First, in Figure~\ref{fig:linear_miscal10} we provide results for the linear miscalibration setting when selecting the top one-tenth of proposals rather than the top one-third of proposals as shown in Figure~\ref{fig:linear_miscal}. We find similar results, with \ouralgorithm\ matching the performance of Swiss NSF and and determinstic with raw mean scores performing worse as miscalibration increases.

\begin{figure}[htbp]
\begin{subfigure}[t]{0.48\textwidth}
\centering
\includegraphics[width=\linewidth]{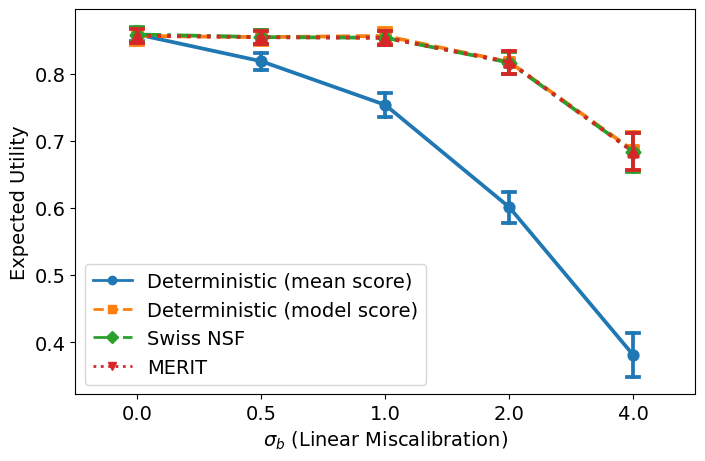}
\caption{Swiss NSF simulated data.}
\label{fig:linear_miscal_swissnsf10}
\end{subfigure}
\hfill
\begin{subfigure}[t]{0.48\textwidth}
\centering
\includegraphics[width=\linewidth]{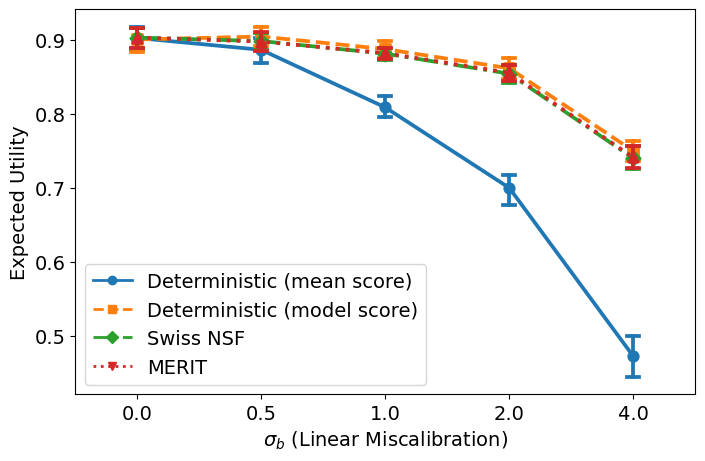}  
\caption{Conference simulated data.}
\label{fig:linear_miscal_conference10}
\end{subfigure}
\caption{Comparison of methods under the Swiss NSF's model of linear miscalibration when choosing $10\%$ of proposals. Error bars show bootstrapped 95\% CIs for the sample mean over 50 samples of randomly generated data from the model.}
\label{fig:linear_miscal10}
\end{figure}

Second, we consider a variant of the linear miscalibration setting where the model is mis-specified. In particular, the simple model used by the Swiss NSF~\cite{heyard2022rethinking} assumes that review score $y_{pr}$ on proposal $p$ from reviewer $r$ is given by $y_{pr} = \theta_p + b_r + \epsilon_{pr}$ where $\theta_p$ is the true quality of the proposal $p$, $b_r$ is the miscalibration of reviewer $r$ and $\epsilon_{pr}$ is random noise. We extend the model, following many prior works~\cite{flach2010kdd,baba2013quality,roos2011calibrate,roos2012statistical}, to include a multiplicative factor per reviewer $a_r$ so that review scores are generated as $y_{pr} = a_r\theta_p + b_r + \epsilon_{pr}$. We assume that $a_r$ is drawn from a $\text{LogNormal}$ distribution with standard deviation $\sigma_a$. Intervals are generated assuming there is no multiplicative factor $a_r$, hence the model used to estimate quality rankings is mis-specified. We show results for this setting, fixing $\sigma_b$ at $1$ and varying $\sigma_a$ when choosing the top one-third of proposals in Figure~\ref{fig:linear_miscal_mult}. As the amount of multiplicative miscalibration increases (so the model gets more prespecified), the performance of all algorithms degrades significantly. However, we still find that \ouralgorithm, Swiss NSF, and deterministic based on model scores perform similarly while raw mean scores performs worse.

\begin{figure}[htbp]
\begin{subfigure}[t]{0.48\textwidth}
\centering
\includegraphics[width=\linewidth]{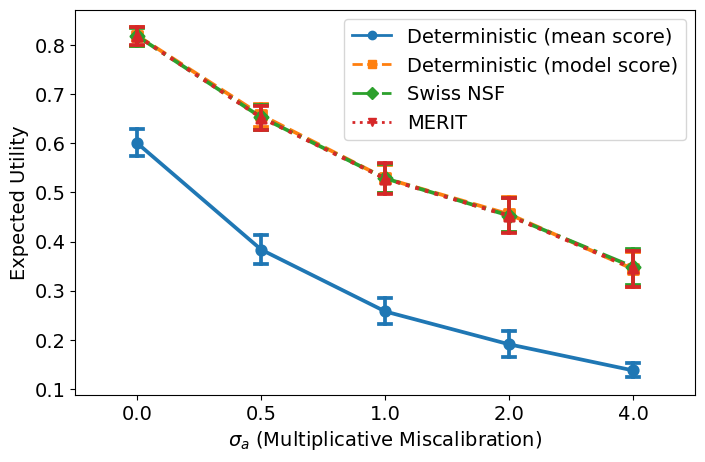}
\caption{Swiss NSF simulated data.}
\label{fig:linear_miscal_swissnsf_mult}
\end{subfigure}
\hfill
\begin{subfigure}[t]{0.48\textwidth}
\centering
\includegraphics[width=\linewidth]{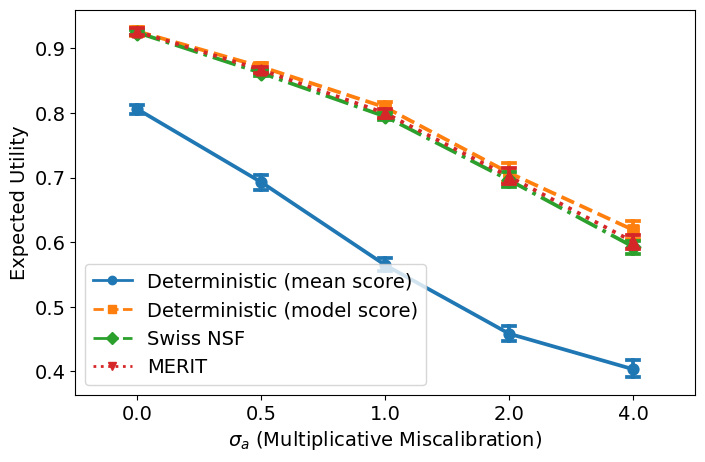}  
\caption{Conference simulated data.}
\label{fig:linear_miscal_conference_mult}
\end{subfigure}
\caption{Comparison of methods under a mis-specified model of linear miscalibration with both multiplicative and additive miscalibration. Error bars show bootstrapped 95\% CIs for the sample mean over 50 samples of randomly generated data from the model.}
\label{fig:linear_miscal_mult}
\end{figure}

\section{Additional Analysis of Computational Runtime Case Studies}
\label{app:additional_data_computation}

In Figure~\ref{fig:runtime_plot}, we show the runtime in seconds of the \ouralgorithm\ algorithm (including all pre-processing and post-processing steps) on each dataset as a function of the acceptance rate. For all methods, the algorithm runs in under five minutes. Runtime increases with acceptance rate, which is expected because the number of constraints grows with the number of selections $\nselected$. We find that the cutting plane algorithm converges in under 30 iterations for all datasets, meaning that it solves under 30 linear programs. Furthermore, the largest linear program solved has 25,000 constraints when choosing $\nselected=5760$ of the $\nproposal=11520$ ICLR papers, suggesting that the cutting plane algorithm scales well in $\nproposal$ and $\nselected$. 

\begin{figure}[htbp]
    \centering
    \includegraphics[width=0.5\linewidth]{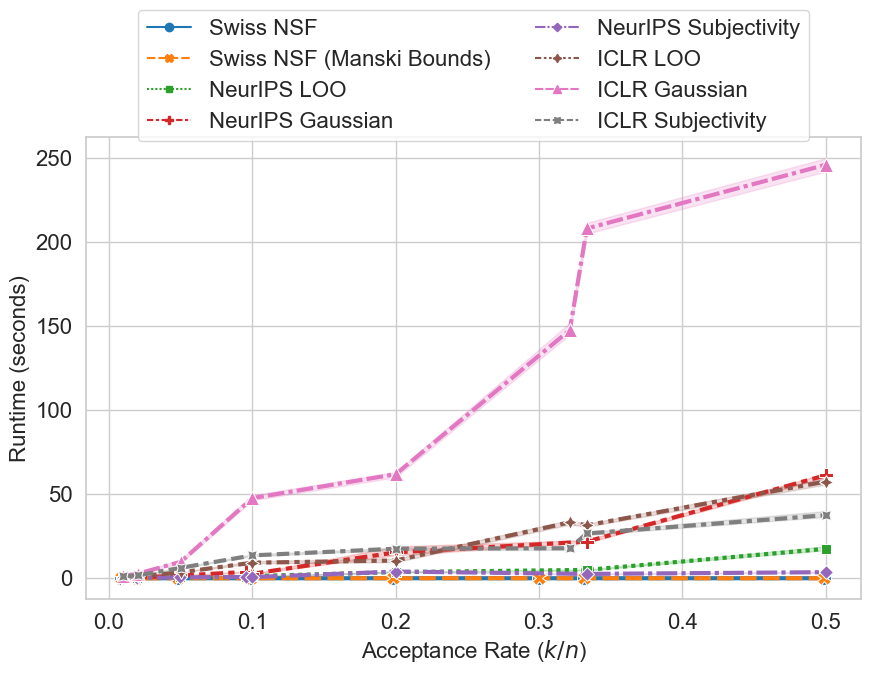}
    \caption{Runtime of \ouralgorithm\ on a standard personal laptop (in seconds) as a function of acceptance rate using review data from the Swiss NSF ($\nproposal=353$), NeurIPS 2024 accepted papers ($\nproposal=4035$) and ICLR 2025 papers ($\nproposal=11520$).}
    \label{fig:runtime_plot}
\end{figure}

In Figure~\ref{fig:cuts_plot}, we present the number of cuts and the number of iterations it takes for the cutting plane algorithm to converge. We note that the size of the LP solved and convergence rate could potentially be optimized further by strategically pruning cuts from the linear program at each iteration, but even without additional optimizations the algorithm yields practical performance.

\begin{figure}[htbp]
    \centering
    \includegraphics[width=0.75\linewidth]{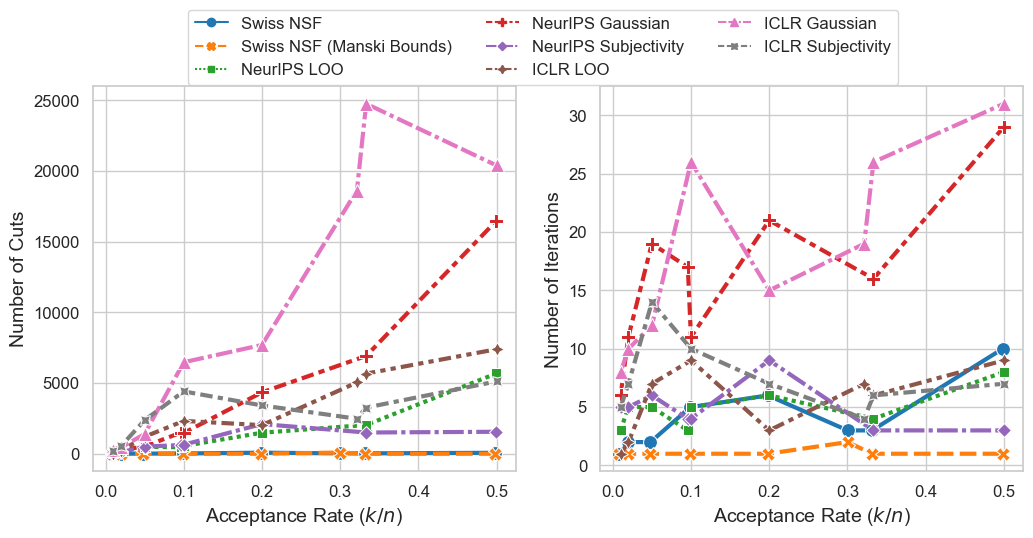}
    \caption{Convergence of cutting plane algorithm as a function of acceptance rate using review data from the Swiss NSF ($\nproposal=353$), NeurIPS 2024 accepted papers ($\nproposal=4035$) and ICLR 2025 papers ($\nproposal=11520$). The number of cuts corresponds to the size of the largest linear program solved in a single iteration of the cutting plane algorithm and the number of iterations corresponds to the total number of LPs solved before convergence.}
    \label{fig:cuts_plot}
\end{figure}

\section{Ex-Ante Optimal Uniform Random Lottery}
\label{app:uniform_lottery}

An additional benefit of \ouralgorithm\ is that it can be easily adapted to handle constraints on the form of the lottery. In some contexts, funders or decision-makers may prefer to implement a uniform lottery, where all candidates subject to randomization are selected with equal probability. This form of randomization may be viewed as simpler, more transparent, and more acceptable to participants than one in which probabilities differ across candidates.

To accommodate such requirements, the optimization problem underlying \ouralgorithm\ can be modified to constrain all randomized candidates to share a common probability of selection. This yields the \emph{ex ante optimal uniform lottery that still satisfies the ex post validity constraints}. The resulting optimization problem is given below:

\begin{align}
    \max_{\prob \in \mathbb{R}^n, \objvalue \in \mathbb{R}, c \in \mathbb{R}} \quad & \objvalue \notag   \\ 
    \text{subject to} \quad 
    & \objvalue \leq \sum_{i=1}^{\nproposal} \prob_i \ind\{\perm(i) \leq \nselected\}, \quad \forall \perm \in \errorpermutations, \label{ref:worst_case_constraints} \\
    & \sum_{i=1}^{\nproposal} \prob_i = \nselected \text{ and } 0 \leq \prob_i \leq 1, \forall i \in [\nproposal] \notag \\ 
    & \prob_i \in \{0,1,c\}, \quad \forall i \in [\nselected] \label{ref:uniform_lottery_constraints}  \\
    & \prob_i = 1 \text { or } \prob_j = 0 \quad \forall i,j \in [\nproposal]: \lcb_i > \ucb_j \label{ref:ex_post_constraints} 
\end{align}

The optimization problem is equivalent to the original optimization problem~\ref{ref:opt_problem}, with the addition of constraint~(\ref{ref:uniform_lottery_constraints}), which forces the lottery to be uniform, and constraint~(\ref{ref:ex_post_constraints}), which ensures ex post validity. These additional constraints turn the linear program into a mixed integer program (MIP). Notably, this MIP has polynomial in $\nproposal$ constraints, excepting the worst-case ordering constraints~(\ref{ref:worst_case_constraints}) present in the original optimization problem. Hence, this can be solved using Algorithm~\ref{alg:cutting_plane} (the \ouralgorithm\ cutting plane algorithm) initialized with the additional integer constraints.

\paragraph{Empirical Performance.}
Figure \ref{fig:comparison_both_models_unif} and Tables \ref{tab:probability_results_uniform}–\ref{tab:case_study_results_uniform} compare the uniform variant of \ouralgorithm\ (\textbf{MERIT Uniform}) against the base \ouralgorithm\ and the Swiss NSF mechanism across datasets. The results show that the uniform constraint has minimal effect on expected utility, while preserving or improving worst-case robustness under our model. Notably, MERIT Uniform achieves nearly identical marginal acceptance probabilities and competitive runtime performance, suggesting that funders can adopt this more interpretable form of randomization without substantial efficiency loss. The runtime of \ouralgorithm\ Uniform can be significantly slower than \ouralgorithm\ due to solving an integer program, but still runs in reasonable time, running in 40 mins instead of 4 mins in the slowest case of ICLR Gaussian, which has over 10,000 candidates from which to choose. Additionally, we performed no additional optimizations of solver parameters in our evaluations of \ouralgorithm\ Uniform, so there may be further speed ups that are possible. 

\begin{table}[!htb]
  \centering
  \begin{tabular}{l|ccc|ccc|ccc}
    \toprule
     & \multicolumn{3}{c|}{\textbf{MERIT}} & \multicolumn{3}{c|}{\textbf{Swiss NSF}} & \multicolumn{3}{c}{\textbf{MERIT Uniform}} \\
    \cmidrule(lr){2-4}\cmidrule(lr){5-7}\cmidrule(lr){8-10}
     Dataset & \% Acc & \% Rand & $p$ & \% Acc & \% Rand & $p$ & \% Acc & \% Rand & $p$ \\
     \midrule
    Swiss NSF & 28.3 & 3.4 & 0.5--0.5 & 28.0 & 3.7 & 0.54 & 28.3 & 3.4 & 0.5 \\
    \midrule
    NeurIPS LOO & 3.4 & 16.5 & 0.36--0.94 & 3.4 & 17.9 & 0.34 & 3.8 & 16.1 & 0.36 \\
    NeurIPS Gaussian & 2.2 & 25.7 & 0.16--0.83 & 4.0 & 27.5 & 0.20 & 4.9 & 19.3 & 0.24 \\
    NeurIPS Subjectivity & 4.5 & 18.7 & 0.14--0.45 & 1.6 & 10.4 & 0.78 & 3.6 & 16.1 & 0.37 \\
    \midrule
    ICLR LOO & 11.1 & 32.4 & 0.51--0.88 & 22.0 & 22.8 & 0.44 & 11.1 & 25.2 & 0.83 \\
    ICLR Gaussian & 9.5 & 34.4 & 0.45--0.87 & 21.2 & 27.7 & 0.40 & 17.9 & 25.6 & 0.56 \\
    ICLR Subjectivity & 17.6 & 24.3 & 0.49--0.88 & 16.2 & 25.2 & 0.63 & 23.0 & 18.4 & 0.49 \\
    \bottomrule
  \end{tabular}
  \caption{Comparison of marginal probabilities of acceptance by MERIT, Swiss NSF, and MERIT Uniform on each dataset. ``Acc'' = guaranteed to be selected ($p = 1$), while ``Rand'' = entered into lottery ($0 < p < 1$). MERIT shows the range of probabilities, while Swiss NSF and MERIT Uniform assign single uniform probabilities.}
  \label{tab:probability_results_uniform}
\end{table}

\begin{table}[!htb]
\centering
\caption{Runtime Comparison for MERIT and MERIT Uniform}
\label{tab:case_study_results_uniform}
\begin{tabular}{lcr@{\hspace{1em}}r}
\toprule
Dataset & & \multicolumn{2}{c}{Runtime (seconds)} \\
\cmidrule{3-4}
 & & MERIT & MERIT Uniform \\
\midrule
Swiss NSF & & 0.046 & 0.043 \\
Swiss NSF (Manski Bounds) & & 0.029 & 0.036 \\
\midrule
NeurIPS LOO & & 0.941 & 1.301 \\
NeurIPS Gaussian & & 5.618 & 43.278 \\
NeurIPS Subjectivity & & 1.883 & 11.319 \\
\midrule
ICLR LOO & & 43.561 & 194.366 \\
ICLR Gaussian & & 242.222 & 2,346.693 \\
ICLR Subjectivity & & 21.060 & 22.225 \\
\bottomrule
\end{tabular}
\end{table}

\begin{figure}[!ht]
  \centering

  \begin{tikzpicture}
    \begin{axis}[
        hide axis,
        xmin=0, xmax=1,
        ymin=0, ymax=1,
        width=0pt,
        height=0pt,
        scale only axis,
        legend style={
            at={(0.5,0.5)},
            anchor=center,
            legend columns=5,
            legend cell align=left,
            draw=none,
            fill=none,
            font=\small,
            column sep=1em
        },
        area legend
    ]
    \addplot+[ybar, draw=black, fill=red!60!black!30, postaction={pattern=crosshatch, pattern color=black}] coordinates {(-1,-1)};
    \addplot+[ybar, draw=black, fill=orange!60, postaction={pattern=horizontal lines, pattern color=black}] coordinates {(-1,-1)};
    \addplot+[ybar, draw=black, fill=green!50!black!30, postaction={pattern=grid, pattern color=black}] coordinates {(-1,-1)};
    \addplot+[ybar, draw=black, fill=purple!40, postaction={pattern=dots, pattern color=black}] coordinates {(-1,-1)};
    \addplot+[ybar, draw=black, fill=blue!40, postaction={pattern=north east lines, pattern color=black}] coordinates {(-1,-1)};
    \legend{Deterministic (mean), Deterministic (model), Swiss NSF, MERIT Uniform, MERIT}
    \end{axis}
  \end{tikzpicture}
  
  \vspace{0.5em} 
  
  \begin{subfigure}[b]{0.42\textwidth}
    \centering
    \begin{tikzpicture}
    \begin{axis}[
        ybar,
        bar width=0.2cm,
        width=5.2cm,
        height=6cm,
        enlarge x limits=0.5,
        ymin=0.8, ymax=0.95,
        ylabel={Expected Utility},
        ylabel style={font=\small},
        symbolic x coords={Swiss NSF, Conference},
        xtick=data,
        xtick style={draw=none},
        x tick label style={rotate=45, anchor=east, font=\small},
        ytick style={draw=none},
        tick label style={font=\small}
    ]
    \addplot+[
        ybar,
        draw=black,
        fill=red!60!black!30,
        postaction={pattern=crosshatch, pattern color=black},
        error bars/.cd,
        y dir=both,
        y explicit,
        error bar style={black}
    ] coordinates {
        (Swiss NSF, 0.857931) +- (0.015597, 0.015597)
        (Conference, 0.889) +- (0.014955, 0.014955)
    };
    \addplot+[
        ybar,
        draw=black,
        fill=orange!60,
        postaction={pattern=horizontal lines, pattern color=black},
        error bars/.cd,
        y dir=both,
        y explicit,
        error bar style={black}
    ] coordinates {
        (Swiss NSF, 0.914310) +- (0.004194, 0.004194)
        (Conference, 0.934) +- (0.004479, 0.004479)
    };
    \addplot+[
        ybar,
        draw=black,
        fill=green!50!black!30,
        postaction={pattern=grid, pattern color=black},
        error bars/.cd,
        y dir=both,
        y explicit,
        error bar style={black}
    ] coordinates {
        (Swiss NSF, 0.914925) +- (0.004113, 0.004113)
        (Conference, 0.933) +- (0.004427, 0.004427)
    };
    \addplot+[
        ybar,
        draw=black,
        fill=purple!40,
        postaction={pattern=dots, pattern color=black},
        error bars/.cd,
        y dir=both,
        y explicit,
        error bar style={black}
    ] coordinates {
        (Swiss NSF, 0.914960) +- (0.004065, 0.004065)
        (Conference, 0.934) +- (0.004438, 0.004438)
    };
    \addplot+[
        ybar,
        draw=black,
        fill=blue!40,
        postaction={pattern=north east lines, pattern color=black},
        error bars/.cd,
        y dir=both,
        y explicit,
        error bar style={black}
    ] coordinates {
        (Swiss NSF, 0.914883) +- (0.004083, 0.004083)
        (Conference, 0.934) +- (0.004438, 0.004438)
    };
    \end{axis}
    \end{tikzpicture}
    \caption{Linear miscalibration model.}
    \label{fig:expected_utility_miscalibration_unif}
  \end{subfigure}
  \hfill
  \begin{subfigure}[b]{0.56\textwidth}
    \centering
    \begin{tikzpicture}
    \begin{axis}[
        ybar,
        bar width=0.2cm,
        width=10.64cm,
        height=6cm,
        enlarge x limits=0.15,
        ylabel={Worst-case Utility},
        ylabel style={font=\small},
        symbolic x coords={
            Swiss NSF,
            NeurIPS~LOO, NeurIPS~Gauss, NeurIPS~Subj,
            ICLR~LOO, ICLR~Gauss, ICLR~Subj
        },
        xtick=data,
        xtick style={draw=none},
        x tick label style={rotate=45, anchor=east, font=\footnotesize},
        ytick style={draw=none},
        tick label style={font=\small},
        ymin=0, ymax=1
    ]
    \addplot+[
        ybar,
        draw=black,
        fill=red!60!black!30,
        postaction={pattern=crosshatch, pattern color=black}
    ] coordinates {
        (Swiss NSF, 0.934)
        (NeurIPS~LOO, 0.054)
        (NeurIPS~Gauss, 0.199)
        (NeurIPS~Subj, 0.377)
        (ICLR~LOO, 0.647)
        (ICLR~Gauss, 0.579)
        (ICLR~Subj, 0.650)
    };
    \addplot+[
        ybar,
        draw=black,
        fill=green!50!black!30,
        postaction={pattern=grid, pattern color=black}
    ] coordinates {
        (Swiss NSF, 0.939)
        (NeurIPS~LOO, 0.371)
        (NeurIPS~Gauss, 0.231)
        (NeurIPS~Subj, 0.347)
        (ICLR~LOO, 0.637)
        (ICLR~Gauss, 0.574)
        (ICLR~Subj, 0.680)
    };
    \addplot+[
        ybar,
        draw=black,
        fill=purple!40,
        postaction={pattern=dots, pattern color=black}
    ] coordinates {
        (Swiss NSF, 0.943)
        (NeurIPS~LOO, 0.394)
        (NeurIPS~Gauss, 0.316)
        (NeurIPS~Subj, 0.473)
        (ICLR~LOO, 0.660)
        (ICLR~Gauss, 0.675)
        (ICLR~Subj, 0.749)
    };
    \addplot+[
        ybar,
        draw=black,
        fill=blue!40,
        postaction={pattern=north east lines, pattern color=black}
    ] coordinates {
        (Swiss NSF, 0.943)
        (NeurIPS~LOO, 0.395)
        (NeurIPS~Gauss, 0.367)
        (NeurIPS~Subj, 0.540)
        (ICLR~LOO, 0.690)
        (ICLR~Gauss, 0.684)
        (ICLR~Subj, 0.749)
    };
    
    
\def\xshiftRedBars{-0.12cm}

\draw[red, thick]
  ([xshift=\xshiftRedBars]axis cs:NeurIPS~Subj,0.347) --
  ([xshift=\xshiftRedBars]axis cs:NeurIPS~Subj,0.473);
\draw[red]
  ([xshift=\xshiftRedBars-2pt]axis cs:NeurIPS~Subj,0.347) --
  ([xshift=\xshiftRedBars+2pt]axis cs:NeurIPS~Subj,0.347);
\draw[red]
  ([xshift=\xshiftRedBars-2pt]axis cs:NeurIPS~Subj,0.473) --
  ([xshift=\xshiftRedBars+2pt]axis cs:NeurIPS~Subj,0.473);
\node[anchor=south,font=\scriptsize,fill=white,inner sep=1pt]
  at ([xshift=\xshiftRedBars]axis cs:NeurIPS~Subj,0.55) {+0.126};

\draw[red, thick]
  ([xshift=\xshiftRedBars]axis cs:ICLR~Gauss,0.574) --
  ([xshift=\xshiftRedBars]axis cs:ICLR~Gauss,0.675);
\draw[red]
  ([xshift=\xshiftRedBars-2pt]axis cs:ICLR~Gauss,0.574) --
  ([xshift=\xshiftRedBars+2pt]axis cs:ICLR~Gauss,0.574);
\draw[red]
  ([xshift=\xshiftRedBars-2pt]axis cs:ICLR~Gauss,0.675) --
  ([xshift=\xshiftRedBars+2pt]axis cs:ICLR~Gauss,0.675);
\node[anchor=south,font=\scriptsize,fill=white,inner sep=1pt]
  at ([xshift=\xshiftRedBars]axis cs:ICLR~Gauss,0.700) {+0.101};

\draw[red, thick]
  ([xshift=\xshiftRedBars]axis cs:NeurIPS~Gauss,0.231) --
  ([xshift=\xshiftRedBars]axis cs:NeurIPS~Gauss,0.316);
\draw[red]
  ([xshift=\xshiftRedBars-2pt]axis cs:NeurIPS~Gauss,0.231) --
  ([xshift=\xshiftRedBars+2pt]axis cs:NeurIPS~Gauss,0.231);
\draw[red]
  ([xshift=\xshiftRedBars-2pt]axis cs:NeurIPS~Gauss,0.316) --
  ([xshift=\xshiftRedBars+2pt]axis cs:NeurIPS~Gauss,0.316);
\node[anchor=south,font=\scriptsize,fill=white,inner sep=1pt]
  at ([xshift=\xshiftRedBars]axis cs:NeurIPS~Gauss,0.38) {+0.085};
    
    \end{axis}
    \end{tikzpicture}
    \caption{Worst-case over interval ordering model (our model).}
    \label{fig:worst_case_utility_unif}
  \end{subfigure}
  \caption{Proportion of top-$\nselected$ proposals selected by different methods with quality data generated under the Swiss NSF model of linear miscalibration and under our model of worst-case over feasible rankings. MERIT Uniform matches performance of algorithms designed for the Swiss NSF's linear model. MERIT uniform recovers much of the gap between MERIT and other methods in the worst-case over intervals defined by our model, as shown by the gaps in NeurIPS~Gaussian, NeurIPS~Subjectivity, and ICLR~Subjectivity.}
  \label{fig:comparison_both_models_unif}
\end{figure}

\end{document}